\documentclass[%
 aip,
 amsmath,amssymb,
reprint, onecolumn %%%%% Comment onecolumn to override one column format %%%%%%%
]{revtex4-1}

\usepackage{graphicx}% Include figure files
\usepackage{dcolumn}% Align table columns on decimal point
\usepackage{bm}% bold math
\usepackage[utf8]{inputenc}
\usepackage[T1]{fontenc}
\usepackage{mathptmx}
\usepackage{etoolbox}

\usepackage{enumerate}
\usepackage{amsmath}
\usepackage{amssymb}
\usepackage{amsfonts}
\usepackage{mathrsfs}
\usepackage{epstopdf}
\usepackage{changepage}
\usepackage{multirow}
\usepackage{booktabs}
\usepackage{lipsum}

\newtheorem{theorem}{Theorem}
\newtheorem{definition}{Definition}

\newtheorem{lemma}{Lemma}
\newtheorem{proposition}{Proposition}
\newtheorem{conjecture}{Conjecture}
\newtheorem{example}{Example}
\newtheorem{corollary}{Corollary}

\def\bcj{\begin{conjecture}}
	\def\ecj{\end{conjecture}}
\def\bcr{\begin{corollary}}
	\def\ecr{\end{corollary}}
\def\bd{\begin{definition}}
	\def\ed{\end{definition}}
\def\bea{\begin{eqnarray}}
	\def\eea{\end{eqnarray}}
\def\bem{\begin{enumerate}}
	\def\eem{\end{enumerate}}
\def\bex{\begin{example}}
	\def\eex{\end{example}}
\def\bim{\begin{itemize}}
	\def\eim{\end{itemize}}
\def\bl{\begin{lemma}}
	\def\el{\end{lemma}}
\def\bma{\begin{bmatrix}}
	\def\ema{\end{bmatrix}}
\def\bpf{\begin{proof}}
	\def\epf{\end{proof}}
\def\bpp{\begin{proposition}}
	\def\epp{\end{proposition}}
\def\bqu{\begin{question}}
	\def\equ{\end{question}}
\def\br{\begin{remark}}
	\def\er{\end{remark}}
\def\bt{\begin{theorem}}
	\def\et{\end{theorem}}

\def\squareforqed{\hbox{\rlap{$\sqcap$}$\sqcup$}}
\def\qed{\ifmmode\squareforqed\else{\unskip\nobreak\hfil
		\penalty50\hskip1em\null\nobreak\hfil\squareforqed
		\parfillskip=0pt\finalhyphendemerits=0\endgraf}\fi}
\def\endenv{\ifmmode\;\else{\unskip\nobreak\hfil
		\penalty50\hskip1em\null\nobreak\hfil\;
		\parfillskip=0pt\finalhyphendemerits=0\endgraf}\fi}
\newenvironment{proof}{\noindent \textbf{{Proof.~} }}{\qed}
\def\Dbar{\leavevmode\lower.6ex\hbox to 0pt
	{\hskip-.23ex\accent"16\hss}D}
\makeatletter
\def\url@leostyle{%
	\@ifundefined{selectfont}{\def\UrlFont{\sf}}{\def\UrlFont{\small\ttfamily}}}
\makeatother
\def\bcj{\begin{conjecture}}
	\def\ecj{\end{conjecture}}
\def\bcr{\begin{corollary}}
	\def\ecr{\end{corollary}}
\def\bd{\begin{definition}}
	\def\ed{\end{definition}}
\def\bea{\begin{eqnarray}}
	\def\eea{\end{eqnarray}}
\def\bem{\begin{enumerate}}
	\def\eem{\end{enumerate}}
\def\bex{\begin{example}}
	\def\eex{\end{example}}
\def\bim{\begin{itemize}}
	\def\eim{\end{itemize}}
\def\bl{\begin{lemma}}
	\def\el{\end{lemma}}
\def\bpf{\begin{proof}}
	\def\epf{\end{proof}}
\def\bpp{\begin{proposition}}
	\def\epp{\end{proposition}}
\def\bqu{\begin{question}}
	\def\equ{\end{question}}
\def\br{\begin{remark}}
	\def\er{\end{remark}}
\def\bt{\begin{theorem}}
	\def\et{\end{theorem}}

\def\btb{\begin{tabular}}
	\def\etb{\end{tabular}}

	\newcommand{\nc}{\newcommand}
	
	\nc{\bbA}{\mathbb{A}} \nc{\bbB}{\mathbb{B}} \nc{\bbC}{\mathbb{C}}
	\nc{\bbD}{\mathbb{D}} \nc{\bbE}{\mathbb{E}} \nc{\bbF}{\mathbb{F}}
	\nc{\bbG}{\mathbb{G}} \nc{\bbH}{\mathbb{H}} \nc{\bbI}{\mathbb{I}}
	\nc{\bbJ}{\mathbb{J}} \nc{\bbK}{\mathbb{K}} \nc{\bbL}{\mathbb{L}}
	\nc{\bbM}{\mathbb{M}} \nc{\bbN}{\mathbb{N}} \nc{\bbO}{\mathbb{O}}
	\nc{\bbP}{\mathbb{P}} \nc{\bbQ}{\mathbb{Q}} \nc{\bbR}{\mathbb{R}}
	\nc{\bbS}{\mathbb{S}} \nc{\bbT}{\mathbb{T}} \nc{\bbU}{\mathbb{U}}
	\nc{\bbV}{\mathbb{V}} \nc{\bbW}{\mathbb{W}} \nc{\bbX}{\mathbb{X}}
	\nc{\bbZ}{\mathbb{Z}}
	
	\nc{\bA}{{\bf A}} \nc{\bB}{{\bf B}} \nc{\bC}{{\bf C}}
	\nc{\bD}{{\bf D}} \nc{\bE}{{\bf E}} \nc{\bF}{{\bf F}}
	\nc{\bG}{{\bf G}} \nc{\bH}{{\bf H}} \nc{\bI}{{\bf I}}
	\nc{\bJ}{{\bf J}} \nc{\bK}{{\bf K}} \nc{\bL}{{\bf L}}
	\nc{\bM}{{\bf M}} \nc{\bN}{{\bf N}} \nc{\bO}{{\bf O}}
	\nc{\bP}{{\bf P}} \nc{\bQ}{{\bf Q}} \nc{\bR}{{\bf R}}
	\nc{\bS}{{\bf S}} \nc{\bT}{{\bf T}} \nc{\bU}{{\bf U}}
	\nc{\bV}{{\bf V}} \nc{\bW}{{\bf W}} \nc{\bX}{{\bf X}}
	\nc{\ba}{{\bf a}} \nc{\be}{{\bf e}} \nc{\bu}{{\bf u}}
	\nc{\brr}{{\bf r}}
	
	\nc{\cA}{{\cal A}} \nc{\cB}{{\cal B}} \nc{\cC}{{\cal C}}
	\nc{\cD}{{\cal D}} \nc{\cE}{{\cal E}} \nc{\cF}{{\cal F}}
	\nc{\cG}{{\cal G}} \nc{\cH}{{\cal H}} \nc{\cI}{{\cal I}}
	\nc{\cJ}{{\cal J}} \nc{\cK}{{\cal K}} \nc{\cL}{{\cal L}}
	\nc{\cM}{{\cal M}} \nc{\cN}{{\cal N}} \nc{\cO}{{\cal O}}
	\nc{\cP}{{\cal P}} \nc{\cQ}{{\cal Q}} \nc{\cR}{{\cal R}}
	\nc{\cS}{{\cal S}} \nc{\cT}{{\cal T}} \nc{\cU}{{\cal U}}
	\nc{\cV}{{\cal V}} \nc{\cW}{{\cal W}} \nc{\cX}{{\cal X}}
	\nc{\cZ}{{\cal Z}}
	
	\nc{\hA}{{\hat{A}}} \nc{\hB}{{\hat{B}}} \nc{\hC}{{\hat{C}}}
	\nc{\hD}{{\hat{D}}} \nc{\hE}{{\hat{E}}} \nc{\hF}{{\hat{F}}}
	\nc{\hG}{{\hat{G}}} \nc{\hH}{{\hat{H}}} \nc{\hI}{{\hat{I}}}
	\nc{\hJ}{{\hat{J}}} \nc{\hK}{{\hat{K}}} \nc{\hL}{{\hat{L}}}
	\nc{\hM}{{\hat{M}}} \nc{\hN}{{\hat{N}}} \nc{\hO}{{\hat{O}}}
	\nc{\hP}{{\hat{P}}} \nc{\hR}{{\hat{R}}} \nc{\hS}{{\hat{S}}}
	\nc{\hT}{{\hat{T}}} \nc{\hU}{{\hat{U}}} \nc{\hV}{{\hat{V}}}
	\nc{\hW}{{\hat{W}}} \nc{\hX}{{\hat{X}}} \nc{\hZ}{{\hat{Z}}}
	
	\nc{\hn}{{\hat{n}}}
	
	\def\dim{\mathop{\rm Dim}}

	\def\rank{\mathop{\rm rank}}
	
	\def\tr{\mathop{\rm Tr}}

	\newcommand{\ket}[1]{|#1\rangle}

	\newcommand{\braket}[2]{\langle#1|#2\rangle}

	\usepackage[
	colorlinks,
	linkcolor = blue,
	citecolor = blue,
	urlcolor = blue]{hyperref}
	\def \qed {\hfill \vrule height7pt width 7pt depth 0pt}
	
\makeatletter
\def\@email#1#2{%
 \endgroup
 \patchcmd{\titleblock@produce}
  {\frontmatter@RRAPformat}
  {\frontmatter@RRAPformat{\produce@RRAP{*#1\href{mailto:#2}{#2}}}\frontmatter@RRAPformat}
  {}{}
}%
\makeatother
\begin{document}

\preprint{AIP/123-QED}

\title[Unextendible and strongly uncompletable product bases]{Unextendible and Strongly Uncompletable Product Bases\\}
% Force line breaks with \\

\author{Xiao-Fan Zhen}

 \affiliation{
School of Mathematical Sciences, Hebei Normal University, Shijiazhuang, 050024, China}
\affiliation{
State Key Laboratory of Networking and Switching Technology, Beijing University of Posts and Telecommunications, Beijing, 100876, China}%Lines break automatically or can be forced with \\

\author{Hui-Juan Zuo}%
 \email{Author to whom correspondence should be addressed: huijuanzuo@163.com}
 
\affiliation{
School of Mathematical Sciences, Hebei Normal University, Shijiazhuang, 050024, China}
\affiliation{
Hebei Key Laboratory of Computational Mathematics and Applications, Shijiazhuang, 050024, China}
\affiliation{
Hebei International Joint Research Center for Mathematics and Interdisciplinary Science, Shijiazhuang, 050024, China}
%\\This line break forced with \textbackslash\textbackslash
%

\author{Fei Shi}
\affiliation{%
Department of Computer Science, School of Computing and Data Science, University of Hong Kong, Hong Kong, 999077, China%\\This line break forced% with \\
}%

\author{Shao-Ming Fei}
\affiliation{School of Mathematical Sciences, Capital Normal University, Beijing, 10048, China}

\date{\today}% It is always \today, today,

\begin{abstract}
In 2003, DiVincenzo {\it et al}. put forward the question that whether there exists an unextendible product basis (UPB) which is an uncompletable product basis (UCPB) in every bipartition
[\href{https://link.springer.com/article/10.1007/s00220-003-0877-6}{DiVincenzo {\it et al}. Commun. Math. Phys. \textbf{238}, 379-410(2003)}]. Recently, Shi {\it et al}. presented a UPB in tripartite systems that is also a strongly uncompletable product basis (SUCPB) in every bipartition [\href{https://iopscience.iop.org/article/10.1088/1367-2630/ac9e14}{Shi {\it et al}. New J. Phys. \textbf{24}, 113-025 (2022)}]. However, whether there exist UPBs that are SUCPBs in only one or two bipartitions remains unknown. We provide a sufficient condition for the existence of SUCPBs based on a quasi U-tile structure. We analyze all possible cases about the relationship between UPBs and SUCPBs in tripartite systems. In particular, we construct a UPB with smaller size $d^3-3d^2+1$ in $\mathbb{C}^{d}\otimes \mathbb{C}^{d}\otimes \mathbb{C}^{d}$, which is an SUCPB in every bipartition and has a smaller cardinality than the existing one.
\end{abstract}

\maketitle

\section{INTRODUCTION}
Quantum nonlocality without entanglement has attracted much attention recently, \cite{  Bennett1999,Walgate2002,Wangyl2015,Zhangzc2015,Xu2016,Zhangzc2016,
Wangyl2017,Zhangzc2017,Halder2018,Jiang2020,Zuo2022,Zhu2022,Zhen2022,Cao2023,Li2023} It has been verified that a set of orthogonal multipartite product (separable) states may be not perfectly distinguished by local operations and classical communication (LOCC).
Such local distinguishability and the smallest number of a set of the locally indistinguishable orthogonal product states have significant applications in
quantum key distributions \cite{Guo2001} and quantum secret sharing. \cite{Rahaman2015,Wang2017,Yang2015}

The construction of locally indistinguishable orthogonal product states is tightly related to the study on unextendible product basis (UPB). A UPB is an incomplete orthogonal product basis whose complementary subspace contains no product states. The entangled state on the complementary subspace to a UPB gives rise a bound entanglement (BE) state. \cite{Augusiak,Chen2011,Johnston2013PRA,Chen,Chen2016,Demianowicz,Nawareg}

In 2003, DiVincenzo {\it et al}. put forward the concept of a strongly uncompletable product basis (SUCPB). \cite{DiVincenzo2003CPM} When the dimension $d$ of the space is even, they constructed a GenTiles$1$ UPB in $\mathbb{C}^d\otimes \mathbb{C}^d$ and a general construction called a GenTiles$2$ UPB in $\mathbb{C}^{d_1}\otimes \mathbb{C}^{d_2}$. They also proposed two open problems: (1) Is there a UPB that is also an uncompletable product basis (UCPB) in every bipartition? (2) Is there a UPB that is still a UPB in every bipartition?
The structure and feature of UPBs have attracted great attention over the past two decades, but less is known for SUCPBs. In 2020, Shi {\it et al}. \cite{ShiPRAUPB} proposed a necessary and sufficient condition for the existence of UPBs in $\mathbb{C}^{d_1}\otimes \mathbb{C}^{d_2}$, that is, a tile structure corresponds to a UPB if and only if the tile structure is a U-tile structure.

Recently, Shi {\it et al}. \cite{Shi2022NJP} showed that there are some UPBs that are SUCPBs in every bipartition for tripartite systems. For multipartite systems containing qubit subsystems, the existence of UPBs with different sizes is of interest. \cite{Bennett1999UPB,DiVincenzo2003CPM,Alon2001,Bravyi,Feng2006,Johnston2013,
Johnston2014,Chen2013,Chen2015,Chen2018JPA,Chen2019QIP,Wang2020QIP,Wang2019} Bennett {\it et al}. \cite{Bennett1999UPB} constructed a Shifts UPB containing four product states in $\mathbb{C}^{2}\otimes \mathbb{C}^{2}\otimes \mathbb{C}^{2}$. Then, a GenShifts UPB with $n+1$ members was proposed by DiVincenzo {\it et al}. in $(\mathbb{C}^{2})^{\otimes n}$ \cite{DiVincenzo2003CPM} when $n$ is odd. By using 1-factorization of complete graphs, Feng presented the minimum size of $4$-qubit UPB containing $6$ product states. \cite{Feng2006} Next, Johnston \cite{Johnston2013} proved that the smallest UPB consists of $11$ states in $(\mathbb{C}^{2})^{\otimes 8}$ and $4k+4$ states in $(\mathbb{C}^{2})^{\otimes 4k}$ for $k\geq 3$. Furthermore, Johnston \cite{Johnston2014} analyzed a complete characterization of all four-qubit UPBs, including the minimal $6$-state UPB and the maximal $12$-state UPB. Wang and Chen \cite{Wang2020QIP} constructed a $7$-qubit UPB of size $10$ and an $8$-qubit UPB of size $18$. In 2021, Wang {\it et al}. \cite{Wang2019} discussed all possible UPBs of size $6$ and $9$ in $\mathbb{C}^{2}\otimes \mathbb{C}^{2}\otimes \mathbb{C}^{4}$. Agrawal {\it et al}. \cite{AgrawalPRA} provided a three-qutrit UPB of size $19$. In Ref.~\cite{ShiJPA}, Shi {\it et al}. generalized the structure with different large sizes in $\mathbb{C}^{d_A}\otimes \mathbb{C}^{d_B}\otimes \mathbb{C}^{d_C}$ and $\mathbb{C}^{d_A}\otimes \mathbb{C}^{d_B}\otimes \mathbb{C}^{d_C}\otimes \mathbb{C}^{d_D}$. In Ref.~\cite{Shi2022NJP}, Shi {\it et al}. proved that these UPBs with strong nonlocality are SUCPBs in every bipartitions. In 2022, Che {\it et al}. \cite{Che2022} provided a strongly nonlocal UPB of size $(d-1)^3+2d+5$ in $\mathbb{C}^d\otimes\mathbb{C}^d\otimes\mathbb{C}^d$ and generalized it to arbitrary tripartite systems.

In the manuscript, we put forward a geometric structure to illustrate SUCPBs. The sufficient condition on the existence of SUCPBs is given by a quasi U-tile structure. We investigate an important problem on whether there exist UPBs that are SUCPBs associated with only one or two bipartitions. We provide two types of UPBs in $\mathbb{C}^d\otimes \mathbb{C}^2\otimes \mathbb{C}^2$ that are SUCPBs in at most one bipartition. Then, we construct a general UPB in $\mathbb{C}^{d}\otimes \mathbb{C}^{d}\otimes \mathbb{C}^{2}$ based on a TILES UPB that is an SUCPB in two bipartitions. For three-qudit systems, we consider UPBs with fewer cardinality that are SUCPBs in all bipartions by optimizing the GenTiles$1$ UPB to construct a UPB with smaller size $d^3-3d^2+1$, which is an SUCPB in every bipartition. We exhibit a systematic analysis on UPBs and SUCPBs in any bipartition.

\section{Quasi U-TILE STRUCTURE}
We first introduce some concepts, facts and notations. For simplicity all the states are assumed to be not normalized. Denote $\{\ket{i}\}_{i\in\bbZ_{d_1}}$ ($\{\ket{j}\}_{j\in\bbZ_{d_2}}$) the computational basis of $\bbC^{d_1}$ ($\bbC^{d_2}$). Then any bipartite state $\ket{\psi}\in\bbC^{d_1}\otimes \bbC^{d_2}$ can be expressed as $\ket{\psi}=\sum_{i\in\bbZ_{d_1}}\sum_{j\in\bbZ_{d_2}}a_{i,j}\ket{i}\ket{j}$, where $\bbZ_d=\{0, 1,\cdots, d-1\}$. A given state $\ket{\psi}$ corresponds to a matrix $M_{d_1\times d_2}=(a_{i,j})_{i\in\bbZ_{d_1},j\in\bbZ_{d_2}}$. $\ket{\psi}$ is a product state if and only if $\rank(M)=1$. For any two states $\ket{\psi_1}$ and $\ket{\psi_2}$ corresponding to matrices $M_1$ and $M_2$, respectively, one has $\braket{\psi_1}{\psi_2}=\tr(M_1^{\dagger}M_2)$. In the following we denote $\text{Sum}(M)$ the sum of all the elements of a matrix $M$.

Let $\cH= \otimes_{i=1}^{n} \mathcal H_i$ be the Hilbert space of an $n$-partite quantum system. An \emph{orthogonal product set (OPS)} $\mathcal S$ is a set of orthogonal product states spanning a proper subspace $\mathcal H_{\mathcal S}$ of $\cH$. \cite{DiVincenzo2003CPM} An \emph{uncompletable product basis (UCPB)} is an OPS whose complementary subspace $\mathcal H_{\mathcal S}^{\bot}$ contains fewer mutually orthogonal product states than the dimension. An \emph{unextendible product basis (UPB)} is a UCPB for which $\mathcal H_{\mathcal S}^{\bot}$ contains no product states. A \emph{strongly uncompletable product basis (SUCPB)} is an OPS spanning a subspace $\mathcal H_{\mathcal S}$ in a locally extended Hilbert space $(\mathcal H_{ext}= \mathcal H_{\mathcal S} \oplus \mathcal H_{\mathcal S}^{\bot})$ such that for all $\mathcal H_{ext}$ the subspace $\mathcal H_{\mathcal S}^{\bot}$ contains fewer mutually orthogonal product states than the dimension. The inclusive relationship among UPBs, SUCPBs and UCPBs is illustrated in Fig.~\ref{fig:rela}.
\begin{figure}[htbp]
 \centering
 \includegraphics[width=5.5cm]{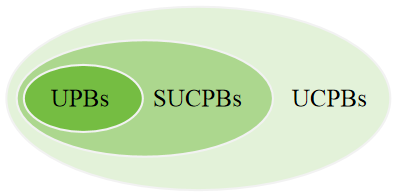}
 \caption{ The Venn diagram of the inclusive relationship among the three concepts: UPBs, SUCPBs and UCPBs.}  \label{fig:rela}
\end{figure}

The tile structure provides an elegant tool for the construction of UPBs. An $m\times n$ tile structure $\cT=\cup_{i=1}^n{t_i}$ is an $m\times n$ rectangle consisting of $n$ tiles, with each tile a subrectangle having row indices $\{r_0, r_1,\cdots,r_{p_i-1}\}$ and column indices $\{c_0, c_1,\cdots , c_{q_i-1}\}$. The row indices or column indices may not be continuous. If $T$ is a subrectangle of $\cT$ consisting of $k$ tiles, $2\leq k\leq n$, then $T$ is called a \emph{special rectangle}. For example, two tile structures $\cT_\mathcal S=\cup_{i=1}^6t_i$ and $\cT_\mathcal U=\cup_{j=1}^5l_j $ are showed in Fig.~\ref{fig:quasi}. The tile $t_3$ of $\cT_\mathcal S$ has row indices $\{0\}$ and column indices $\{0,1,2,3\}$. $T=t_5\cup t_6$ is a special rectangle of $\cT_\mathcal S$.
\begin{figure}[htbp]
 \centering
 \includegraphics[width=7cm]{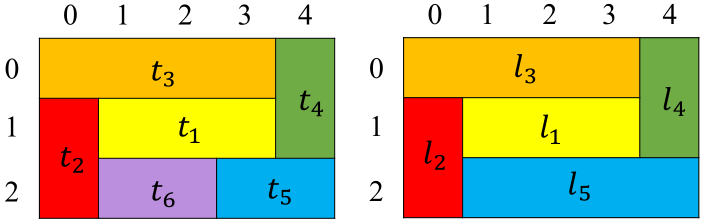}
 \caption{The quasi U-tile structure  $\mathcal T_{\mathcal S}=\cup _{i=1}^{6} t_{i}$ and the U-tile structure $\mathcal T_{\mathcal U}=\cup _{j=1}^{5} l_{j}$.}  \label{fig:quasi}
\end{figure}

Next, we introduce a U-tile structure and recall the relationship between a UPB and a U-tile structure. Given a tile structure $\mathcal T$, if any special rectangle $T$ of $\cT$ cannot be partitioned into two smaller special rectangles or tiles of $\cT$, then $\mathcal T$ is said to be a \emph{U-tile structure}. \cite{ShiPRAUPB} For example, the tile structure $\cT_\mathcal U$ in Fig.~\ref{fig:quasi} is a U-tile structure, because it has only one special rectangle $T=\cT$. However, $\cT_\mathcal S$ in Fig.~\ref{fig:quasi} is not a U-tile structure due to the special rectangle $t_5\cup t_6$.

A tile structure with $s$-tiles corresponds to a UPB of size $(mn-s+1)$ in $\mathbb{C}^m\otimes \mathbb{C}^n$ if and only if this tile structure is a U-tile structure \cite{ShiPRAUPB}.
Let $\mathcal S$ be an OPS in $\mathcal H= \otimes_{i=1}^{n} \mathcal H_i$. If all product states in $\mathcal H_{\mathcal S}^{\bot}$ cannot span $\mathcal H_{\mathcal S}^{\bot}$, then $\mathcal S$ is an SUCPB. \cite{Shi2022NJP}

Based on the tile structure, we consider a sufficient condition for an OPS to be an SUCPB. In order to illustrate the geometric characterization of the tile structure corresponding to an SUCPB, we introduce a quasi U-tile structure. We say that a tile structure $\mathcal T_{\mathcal S}=\cup _{i=1}^{n} t_{i}$ is called a \emph{quasi U-tile structure} if there exists a partition $\{t_i\}_{i=1}^n$ into $m$ subsets $\{l_j\}_{j=1}^m$ ($m\geq 5$), where $l_j=\cup_{s=1}^p t_{k_{s}}$, such that
\begin{enumerate}[(i)]
\item \label{item1} $l_j$ is a new tile which cannot be extended to a larger tile by adding other tiles except for $\cT_{\cS}$;
\item \label{item2}  The tiles $\{t_{k_{s}}\}_{s=1}^p$ have the same row indices or column indices;
\item \label{item3}  The new tile structure $\mathcal T_{\mathcal U}=\cup _{j=1}^{m} l_{j}$ is a U-tile structure.
\end{enumerate}

As an example, consider the tile structure $\mathcal T_{\mathcal S}=\cup _{i=1}^{6} t_{i}$ in Fig.~\ref{fig:quasi}. Let the new tiles $l_i=t_i$ for $1\leq i\leq 4$ and $l_5=t_5\cup t_6$. Then each $l_i$ satisfies the conditions (\ref{item1}) and (\ref{item2}). The new tile structure $\mathcal T_{\mathcal U}=\cup _{j=1}^{5} l_{j}$ is a U-tile structure by definition. Therefore, the tile structure $\mathcal T_{\mathcal S}=\cup _{i=1}^{6} t_{i}$ is a quasi U-tile structure.

\begin{lemma}\label{lem:quasi_sucpb} If a $d_1\times d_2$ tile structure with $n$-tiles is a quasi U-tile structure, then the quasi U-tile structure corresponds to an SUCPB of size $d_{1}d_{2}-n+1$ in $\mathbb{C}^{d_1}\otimes \mathbb{C}^{d_2}$.
\end{lemma}

\begin{proof}
Consider a $d_1\times d_2$ quasi U-tile structure $\mathcal T_{\mathcal S}=\cup _{i=1}^{n} t_{i}$, where each tile $t_i$ has row indices $\{r_0,$ ~$r_1,\cdots , r_{p_i-1}\}$ and column indices $\{c_0, c_1,\cdots , c_{q_i-1}\}$. An OPS from the tile $t_i$ is
\begin{equation}
\begin{aligned}
\mathcal{A}_{i}=\{|\phi_{i}^{(k,l)}\rangle= (\sum_{e\in \mathbb Z_{p_i}}w_{p_i}^{ke}|r_e\rangle)(\sum_{e\in \mathbb Z_{q_i}}w_{q_i}^{le}|c_e\rangle)\mid (k,l)\in \bbZ_{p_i}\times \bbZ_{q_i}\},
\end{aligned}
\end{equation}
where $w_x=e^{\frac{2\pi i}{x}}$ for any $x$.
The set $\cB=\cup _{i=1}^{n} \cA_{i}$ is an OPB in $\bbC^{d_1}\otimes \bbC^{d_2}$. The ``stopper state" is defined as
\begin{equation}
|S\rangle = (\sum_{i\in \mathbb Z_{d_{1}}}|i\rangle)\otimes(\sum_{j\in \mathbb Z_{d_{2}}}|j\rangle).
\end{equation}
We claim that $\mathcal U_{\mathcal S}= \cup _{i=1}^{n}(\mathcal{A}_{i}\backslash \{|\phi_{i}^{(0,0)}\rangle\})\cup \{|S\rangle\}$ is an SUCPB of size $d_{1}d_{2}-n+1$.

Let $\cC_1=\cup _{i=1}^{n}(\mathcal{A}_{i}\backslash \{|\phi_{i}^{(0,0)}\rangle\})$ and $\cC_2=\cup_{i=1}^{n} \{|\phi_{i}^{(0,0)}\rangle\}$. Then $\cH_{\cC_1}\cup \cH_{\cC_2}=\cB$. Since $\cH_{\cC_1}\subset \cU_{\cS}$, $\cH_{\mathcal U_{\mathcal S}}^{\bot}\subset \cH_{\cC_1}^{\bot}=\cH_{\cC_2}$. Let $|\phi\rangle\in \mathcal H_{\mathcal U_{\mathcal S}}^{\bot}$ be any product state, $\ket{\phi}=\sum_{i=1}^na_i\ket{\phi_i(0,0)}$, where
$a_{i}\in \mathbb C$, $1\leq i\leq n$. The state $\ket{\phi}$ corresponds to a matrix $M_{d_1\times d_2}$, which has a similar structure to the quasi U-tile structure $\mathcal T_{\mathcal S}$, i.e., the position of the tile $t_{i}$ in a quasi U-tile structure $\mathcal T_{\mathcal S}$ corresponds exactly to the same entry of the matrix $M_{d_1\times d_2}$. For example, $\mathcal T_{\mathcal S}=\cup _{i=1}^{6} t_{i}$ is a quasi U-tile structure in Fig.~\ref{fig:quasi}. The state $\ket{\phi}$ corresponds to the following matrix
\begin{equation}
M_{3\times 5}=\begin{pmatrix}
a_3  &a_3  &a_3  &a_3  &a_4\\
a_2  &a_1  &a_1  &a_1  &a_4\\
a_2  &a_6  &a_6  &a_5  &a_5\\
\end{pmatrix},
\end{equation}
which has a similar structure to $\mathcal T_{\mathcal S}$.

Note that the stopper state $\ket{S}$ corresponds to a all-ones matrix $J$. Since $\braket{\phi}{S}=0$, we have the equation $\tr(M_{d_1\times d_2}^{\dagger}J)=0$. Hence, $\text{Sum}(M_{d_1\times d_2})=0$. Moreover, since $|\phi\rangle$ is a product state, we have $\rank(M_{d_1\times d_2})=1$.

By the definition  of quasi U-tile, there is a partition of $\{t_i\}_{i=1}^{n}$ into $m$ ($m\geq 5$) subsets $\{l_j\}_{j=1}^m$. Without loss of generality, we assume that
 \begin{equation}
\begin{aligned}
l_1&=t_{{k_0}+1}\cup t_{{k_0}+2}\cup \cdots \cup t_{k_1},\\
l_2&=t_{{k_1}+1}\cup t_{{k_1}+2}\cup \cdots \cup t_{k_2}, \\
&~~~~~~~~\cdots~\cdots~\cdots\\
l_m&=t_{{k_{m-1}}+1}\cup t_{{k_m-1}+2}\cup \cdots \cup t_{k_m},
\end{aligned}
\end{equation}
where $k_0=0$ and $k_m=n$.

From the conditions (\ref{item1})-(\ref{item3}), if $\text{Sum}(M_{d_1\times d_2})=0$ and $\rank(M_{d_1\times d_2})=1$, there are only
$m$ cases for the matrix $M_{d_1\times d_2}$, that is,
%\begin{small}
\begin{equation}
\begin{aligned}
\sum \limits_{s=k_{0}+1}^{k_1}p_{s}q_{s}\cdot a_{s}&=0~\text{and} \ a_s=0  \ \text{for} \ s\notin\{k_{0}+1, \ldots , k_1\},\\
\sum \limits_{s=k_{1}+1}^{k_2}p_{s}q_{s}\cdot a_{s}&=0~\text{and} \ a_s=0  \ \text{for} \ s\notin\{k_{1}+1, \ldots , k_2\},\\
&~~~~~~~~\cdots~\cdots~\cdots\\
\sum \limits_{s=k_{m-1}+1}^{k_m}p_{s}q_{s}\cdot a_{s}&=0~\text{and} \ a_s=0  \ \text{for} \ s\notin\{k_{m-1}+1, \ldots , k_m\}.
\end{aligned}
\end{equation}
%\end{small}
It means that the product state $|\phi\rangle$ must belong to one of the $m$  subspaces
\begin{equation}
\begin{aligned}
O_{1}=&\left\{\sum \limits_{s=k_{0}+1}^{k_1}a_{s}|\phi_{s}^{(0,0)}\rangle\mid \sum \limits_{s=k_{0}+1}^{k_1}p_{s}q_{s}\cdot a_{s}=0\right\},\\
O_{2}=&\left\{\sum \limits_{s=k_{1}+1}^{k_2}a_{s}|\phi_{s}^{(0,0)}\rangle\mid \sum \limits_{s=k_{1}+1}^{k_2}p_{s}q_{s}\cdot a_{s}=0\right\},\\
&~~~~~~~~~~~~~~~~\cdots~\cdots~\cdots\\
O_{m}=&\left\{\sum \limits_{s=k_{m-1}+1}^{k_m}a_{s}|\phi_{s}^{(0,0)}\rangle\mid \sum \limits_{s=k_{m-1}+1}^{k_m}p_{s}q_{s}\cdot a_{s}=0\right\},
\end{aligned}
\end{equation}
where $\dim (O_{1})=k_{1}-k_{0}-1$, $\dim (O_{2})=k_{2}-k_{1}-1$, $\cdots$, $\dim (O_{m})=k_{m}-k_{m-1}-1$. Since the subspaces $\{O_{i}\}_{i=1}^{m}$ are mutually orthogonal, $\dim (O_{1}+O_{2}+\cdots +O_{m})=k_{m}-k_{0}-m=n-m$. Moreover, since $\dim (\mathcal H_{\mathcal U_{\mathcal S}}^{\bot})=n-1$, one has $n-m\leq n-5<n-1$. Hence, all product states in $\cH_{\mathcal U_{\mathcal S}}^{\bot}$ cannot span $\mathcal H_{\mathcal U_{\mathcal S}}^{\bot}$. Therefore, $\mathcal U_{\mathcal S}$ is an SUCPB of size $d_{1}d_{2}-n+1$ by using conclusion in Ref. \cite{Shi2022NJP}. This completes the proof.
\end{proof}

\begin{example}
Since $\mathcal T_{\mathcal S}=\cup _{i=1}^{6} t_{i}$ in Fig.~\ref{fig:quasi} is a $3\times 5$ quasi U-tile structure, we have an SUCPB of size $10$ in $\mathbb{C}^{3}\otimes \mathbb{C}^{5}$ by Lemma~\ref{lem:quasi_sucpb} as follows:
\begin{equation}\label{eq:35}
\begin{aligned}
&|\phi_{1}^{(1)}\rangle=|1\rangle(|1\rangle+\omega_{3}^{1}|2\rangle +\omega_{3}^{2}|3\rangle),~~~~~|\phi_{2}^{(1)}\rangle=|1-2\rangle|0\rangle,\\
&|\phi_{1}^{(2)}\rangle=|1\rangle(|1\rangle+\omega_{3}^{2}|2\rangle +\omega_{3}^{1}|3\rangle),~~~~~|\phi_{4}^{(1)}\rangle=|0-1\rangle|4\rangle,\\
&|\phi_{3}^{(1)}\rangle=|0\rangle(|0\rangle+|1\rangle -|2\rangle-|3\rangle),~~~~|\phi_{5}^{(1)}\rangle=|2\rangle|3-4\rangle,\\
&|\phi_{3}^{(2)}\rangle=|0\rangle(|0\rangle-|1\rangle +|2\rangle-|3\rangle),~~~~|\phi_{6}^{(1)}\rangle=|2\rangle|1-2\rangle,\\
&|\phi_{3}^{(3)}\rangle=|0\rangle(|0\rangle-|1\rangle -|2\rangle+|3\rangle),~~~~|S\rangle=|0+1+2\rangle|0+1+2+3+4\rangle.
\end{aligned}
\end{equation}
There is only one product state $|\phi \rangle=|2\rangle|1+2-3-4\rangle$ in the complementary subspace of the SUCPB.
\end{example}

To prove that an orthogonal product set ${\mathcal S}$ is an SUCPB, we only need to show that the tile structure corresponding to ${\mathcal S}$ is a quasi U-tile structure by Lemma~\ref{lem:quasi_sucpb}.

A U-tile structure in bipartite systems can be generalized to multipartite systems Ref.~\cite{ShiPRAUPB}. A $d_1\times d_2\times d_3$ tile structure $\cC=\cup_{i=1}^n{t_i}$  is a $d_1\times d_2\times d_3$ cube consisting of $n$ tiles, where each tile $t_i$ is a subcube with length indices $\mathcal L_i$, width indices $\mathcal W_i$ and height indices $\mathcal H_i$. A \emph{special cube} $C$ of $\cC$ consists of at least two tiles. Similarly, if for any special cube $C$ of $\cC=\cup_{i=1}^n{t_i}$ ($n\geq 5$), $C\neq R\cup S$, where $R$ ($S$) is a special cube or a tile, then $\mathcal C$ is a U-tile structure. A $d_1\times d_2\times d_3$ U-tile structure with $n$-tiles corresponds to a UPB of size $d_{1}d_{2}d_{3}-n+1$ in $\mathbb{C}^{d_{1}}\otimes \mathbb{C}^{d_{2}}\otimes \mathbb{C}^{d_{3}}$.

\section{UPBS THAT ARE SUCPBS IN TWO
BIPARTITIONS}\label{sec:sucpbtwo}

We first present a UPB in $\mathbb{C}^{3}\otimes \mathbb{C}^{3}\otimes \mathbb{C}^{2}$ from the U-tile structure in Fig.~\ref{fig:332} and generalize the structure to $\mathbb{C}^{d}\otimes \mathbb{C}^{d}\otimes \mathbb{C}^{2}$. Denote $A$, $B$, and $C$ the first, second and third subsystems of a tripartite system. Since an OPS in $\mathbb{C}^{2}\otimes \mathbb{C}^{n}$ can be extended to an OPB, \cite{Bennett1999UPB,DiVincenzo2003CPM} the UPB in bipartition $C$ and $AB$, $C|AB$, is a completable orthogonal product basis.

Before we prove that the UPB is an SUCPB in $A|BC$ and $B|CA$ bipartitions, we first present a UPB in $\mathbb{C}^{3}\otimes \mathbb{C}^{3}\otimes \mathbb{C}^{2}$ based on the TILES UPB given by Bennett {\it et al}. \cite{Bennett1999UPB}

\begin{proposition}
In $\mathbb{C}^{3}\otimes \mathbb{C}^{3}\otimes \mathbb{C}^{2}$, the following UPB of size $10$ is an SUCPB in A$|$BC and B$|$CA bipartitions:
\begin{equation}\label{eq:332}
\begin{aligned}
|\phi_{0}\rangle=&|1\rangle_{A}|1\rangle_{B}|0-1\rangle_{C},~~~~~~~~~
|\phi_{5}\rangle=|0-1\rangle_{A}|0\rangle_{B}|1\rangle_{C},\\
|\phi_{1}\rangle=&|0\rangle_{A}|0-1\rangle_{B}|0\rangle_{C},~~~~~~~~~
|\phi_{6}\rangle=|0\rangle_{A}|1-2\rangle_{B}|1\rangle_{C},\\
|\phi_{2}\rangle=&|0-1\rangle_{A}|2\rangle_{B}|0\rangle_{C},~~~~~~~~~
|\phi_{7}\rangle=|1-2\rangle_{A}|2\rangle_{B}|1\rangle_{C},\\
|\phi_{3}\rangle=&|2\rangle_{A}|1-2\rangle_{B}|0\rangle_{C},~~~~~~~~~
|\phi_{8}\rangle=|2\rangle_{A}|0-1\rangle_{B}|1\rangle_{C},\\
|\phi_{4}\rangle=&|1-2\rangle_{A}|0\rangle_{B}|0\rangle_{C},~~~~~~~~~
|\phi_{9}\rangle=|0+1+2\rangle_{A}|0+1+2\rangle_{B}|0+1\rangle_{C}.
\end{aligned}
\end{equation}
\end{proposition}

\begin{proof} With respect to Fig.~\ref{fig:332} we denote $\mathcal V^1= \cup _{i=0}^{9}|\phi_{i}\rangle$ the UPB given by Eq. \eqref{eq:332}, referred to as $\mathcal V_{A|BC}^{1}$ or $\mathcal V_{B|CA}^{1}$ in $A|BC$ or $B|CA$ bipartition. We construct the tile structure with $9$-tiles corresponding to $\mathcal V_{A|BC}^{1}$ and $\mathcal V_{B|CA}^{1}$ in Fig.~\ref{fig:36}, denoted as $\mathcal T_{\mathcal V_{A|BC}^{1}}$ and $\mathcal T_{\mathcal V_{B|CA}^{1}}$, respectively. In order to prove that the UPB $\mathcal V^1$ is an SUCPB in two bipartitions, we demonstrate that the tile structures $\mathcal T_{\mathcal V_{A|BC}^{1}}$ and $\mathcal T_{\mathcal V_{B|CA}^{1}}$ are quasi U-tile ones.
\begin{figure}[htbp]
 \centering
 \begin{minipage}[t]{0.48\textwidth}
 \centering
 \includegraphics[width=6.5cm]{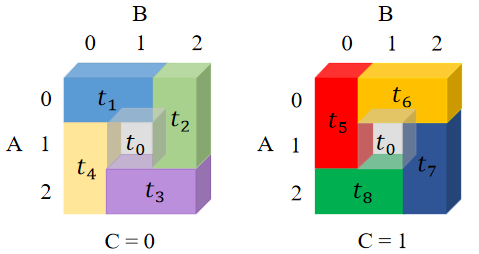}
 \caption{The $3\times 3\times 2$ U-tile structure $\mathcal T_{\mathcal V^{1}}$ with $9$-tiles. }  \label{fig:332}
 \end{minipage}
 \begin{minipage}[t]{0.48\textwidth}
 \centering
 \includegraphics[width=9.0cm]{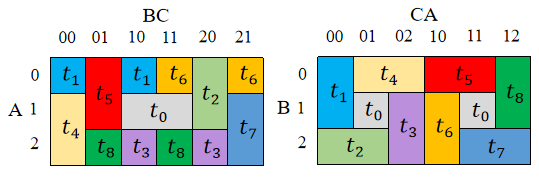}
 \caption{The quasi U-tile structures of $\mathcal T_{\mathcal V_{A|BC}^{1}}$ and $\mathcal T_{\mathcal V_{B|CA}^{1}}$.}  \label{fig:36}
 \end{minipage}
\end{figure}

The tile structure $\mathcal T_{\mathcal V_{A|BC}^{1}}$ has a partition of $\{t_i\}_{i=0}^{8}$ into $5$ subsets $\{l_j\}_{j=0}^4$, that is, $l_0=t_{0}$, $l_1=t_{1}\cup t_{6}$, $l_2=t_{2}\cup t_{5}$, $l_3=t_{3}\cup t_{8}$ and $l_4=t_{4}\cup t_{7}$. Under this partition, every new tile $l_j$ satisfies the conditions (\ref{item1}) and (\ref{item2}). The new tile structure $\mathcal T_{\mathcal V_{A|BC}^{1'}}=\cup _{j=1}^{5} l_{j}$ is a U-tile structure, which means that the condition (\ref{item3}) is satisfied. Thus, the tile structure $\mathcal T_{\mathcal V_{A|BC}^{1}}$ is a quasi U-tile structure according to definition.

Similarly, the tile structure $\mathcal T_{\mathcal V_{B|CA}^{1}}$ has a partition of $\{t_i\}_{i=0}^{8}$ into $5$ subsets $\{l_j\}_{j=0}^4$, i.e., $l_0=t_{0}$, $l_1=t_{1}\cup t_{8}$, $l_2=t_{2}\cup t_{7}$, $l_3=t_{3}\cup t_{6}$ and $l_4=t_{4}\cup t_{5}$. Obviously, the new tile structure $\mathcal T_{\mathcal V_{B|CA}^{1'}}=\cup _{j=1}^{5} l_{j}$ is a U-tile structure. Further, the tile structure $\mathcal T_{\mathcal V_{B|CA}^{1}}$ is a quasi U-tile structure. The sets $\mathcal V_{A|BC}^{1}$ and $\mathcal V_{B|CA}^{1}$ are SUCPBs.
\end{proof}

Next, we propose the decomposition of $d\times d \times 2$ tile structure for $d\geq 3$. Consider a $d\times d \times 2$ tile structure with $k$ layers, each of which is partitioned into $8$ tiles, where $1\leq k\leq \lceil \frac{d-2}{2}\rceil$. Note that the outermost layer is labeled by $k=1$ and the innermost layer is labeled by $k=\lceil \frac{d-2}{2}\rceil$. Adding the center tile, we obtain the decomposition of the $d\times d \times 2$ tile structure with $(8\lceil \frac{d-2}{2}\rceil+1)$-tiles. For example, when $d$ is odd, the tile structure in Fig.~\ref{fig:dd2} is a $d\times d \times 2$ U-tile structure.
\begin{figure*}
 \centering
 \includegraphics[width=11.5cm]{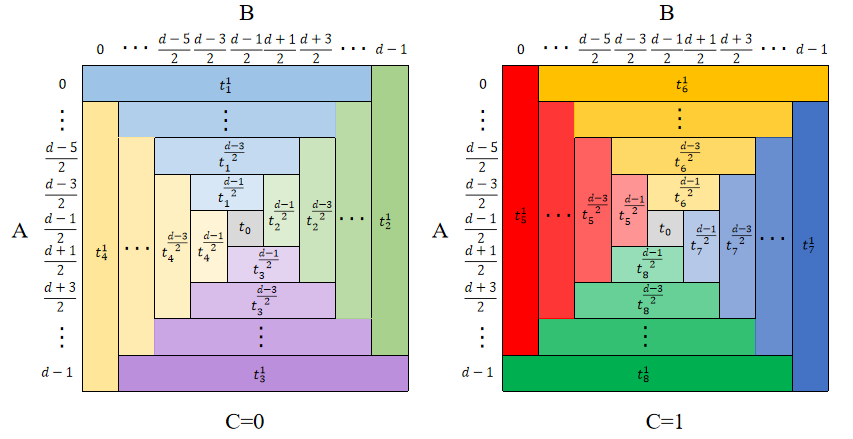}
 \caption{For odd $d$, the $d\times d\times 2$ U-tile structure $\mathcal T_{\mathcal U^1}$ with $4d-3$-tiles is inscribed with two $d\times d$ plane structures.}  \label{fig:dd2}
\end{figure*}

For the case of $\mathbb{C}^{d}\otimes \mathbb{C}^{d}\otimes \mathbb{C}^{2}$, we give the explicit forms of the sets $\mathcal{A}_{i}^{k}$ corresponding to the tiles $t_{i}^{k}=\mathcal L_{i}^{k}\times \mathcal W_{i}^{k}\times \mathcal H_{i}^{k}$, taking from each layer with $1\leq i\leq 8$ and $1\leq k\leq \lceil \frac{d-2}{2}\rceil$:
\begin{equation}
\begin{aligned}
\mathcal{A}_{1}^{k}&=\{|l-1\rangle_{A}|\alpha _{i}\rangle _{B}|0\rangle_{C}\mid i\in \mathbb Z_{d-2l}\backslash \{0\}\},~~~
\mathcal{A}_{2}^{k}=\{|\alpha _{i}\rangle _{A}|d-l\rangle_{B}|0\rangle_{C}\mid i\in \mathbb Z_{d-2l}\backslash \{0\}\},\\
\mathcal{A}_{3}^{k}&=\{|d-l\rangle_{A}|\beta _{i}\rangle _{B}|0\rangle_{C}\mid i\in \mathbb Z_{d-2l}\backslash \{0\}\},~~~
\mathcal{A}_{4}^{k}=\{|\beta _{i}\rangle _{A}|l-1\rangle_{B}|0\rangle_{C}\mid i\in \mathbb Z_{d-2l}\backslash \{0\}\},\\
\mathcal{A}_{5}^{k}&=\{|\alpha _{i}\rangle _{A}|l-1\rangle_{B}|1\rangle_{C}\mid i\in \mathbb Z_{d-2l}\backslash \{0\}\},~~~
\mathcal{A}_{6}^{k}=\{|l-1\rangle_{A}|\beta _{i}\rangle _{B}|1\rangle_{C}\mid i\in \mathbb Z_{d-2l}\backslash \{0\}\},\\
\mathcal{A}_{7}^{k}&=\{|\beta _{i}\rangle _{A}|d-l\rangle_{B}|1\rangle_{C}\mid i\in \mathbb Z_{d-2l}\backslash \{0\}\},~~~
\mathcal{A}_{8}^{k}=\{|d-l\rangle_{A}|\alpha _{i}\rangle _{B}|1\rangle_{C}\mid i\in \mathbb Z_{d-2l}\backslash \{0\}\},\\
\end{aligned}
\end{equation}
where $|\alpha _{i}\rangle _{X}=\sum \limits_{s=0}^{d-2k}w_{d-2k+1}^{is}|s+k-1\rangle_{X}$ and $|\beta _{i}\rangle _{X}=\sum \limits_{s=0}^{d-2k}w_{d-2k+1}^{is}|s+k\rangle_{X}$ for $X\in \{A, B, C\}$. In particular, the center tile $t_{0}$ ($t_{0'}$) corresponds to the set $\mathcal{A}_{0}$ ($\mathcal{A}_{0'}$) when $d$ is odd (even):
\begin{equation}
\begin{aligned}
\mathcal{A}_{0}=&\{|\frac{d-1}{2}\rangle_{A}|\frac{d-1}{2}\rangle _{B}|\eta _{k}\rangle_{C}\mid k\in \mathbb{Z}_{2}\backslash \{0\}\},\\
\mathcal{A}_{0'}=&\{|\xi _{i}\rangle_{A}|\xi _{j}\rangle_{B}|\eta _{k}\rangle_{C}\mid (i,j,k)\in \mathbb Z_{2}\times \mathbb Z_{2}\times \mathbb Z_{2}\backslash \{(0,0,0)\}\},\\
\end{aligned}
\end{equation}
where $|\eta _{s}\rangle_{X}=|0\rangle_{X}+(-1)^{s}|1\rangle_{X}$, $|\xi _{s}\rangle_{X}=|\frac{d}{2}-1\rangle_{X}+(-1)^{s}|\frac{d}{2}\rangle_{X}$ for $X\in \{A, B, C\}$, $s\in \mathbb Z_{2}$.

We put forward a UPB of size $2d^2-4d+4$ corresponding to a $d\times d\times 2$ U-tile structure with $4d-3$-tiles in Fig.~\ref{fig:dd2} when $d$ is odd.

\begin{theorem} The set $\cup _{i=1}^{8}\mathcal{A}_{i}^{k}\cup \mathcal{A}_{0}\cup \{|S\rangle\}$ in $\mathbb{C}^{d}\otimes \mathbb{C}^{d}\otimes \mathbb{C}^{2}$, denoted as $\mathcal U^1$, is a UPB of size $2d^2-4d+4$ for odd $d$, $d\geq 3$ and $1\leq k\leq \frac{d-1}{2}$. It is an SUCPB in A$|$BC and B$|$CA bipartitions.
\end{theorem}

\begin{proof} Denote the UPB $\mathcal U^1$ in $A|BC$ and $B|CA$ bipartitions as $\mathcal U_{A|BC}^{1}$ and $\mathcal U_{B|CA}^{1}$, respectively. To show that the sets $\mathcal U_{A|BC}^{1}$ and $\mathcal U_{B|CA}^{1}$ are SUCPBs, we show that the tile structures corresponding to $\mathcal U_{A|BC}^{1}$ and $\mathcal U_{B|CA}^{1}$ (denoted by $\mathcal T_{\mathcal U_{A|BC}^{1}}$ and $\mathcal T_{\mathcal U_{B|CA}^{1}}$) are quasi U-tile structures.

For the tile structure $T_{\mathcal U_{A|BC}^{1}}$ in Fig.~\ref{fig:d2d}, since the tiles $t_{1}^{k}$ and $t_{6}^{k}$ have the same row indices $k-1$, they can be combined to construct the new tiles $l_{1}^{k}=t_{1}^{k}\cup t_{6}^{k}$. Similarly, the other new tiles $l_{j}^{k}$ can be obtained as follows:
$l_{0}=t_{0}$, $l_{2}^{k}=t_{2}^{k}\cup t_{5}^{k}$, $l_{3}^{k}=t_{3}^{k}\cup t_{8}^{k}$ and $l_{4}^{k}=t_{4}^{k}\cup t_{7}^{k}$, where $1\leq k\leq \frac{d-1}{2}$. The partition to form a new tile structure with $(2d-1)$-tiles, denoted by $\mathcal T_{\mathcal U_{A|BC}^{1'}}=\cup_{j,k}\{l_{j}^{k}\}$, is a U-tile structure. By definition $\mathcal T_{\mathcal U_{A|BC}^{1}}$ is a quasi U-tile structure.

According to the tile structure $\mathcal T_{\mathcal U_{B|CA}^{1}}$, see Fig.~\ref{fig:d2d}, we have a new tile structure with $(2d-1)$-tiles, denoted as $\mathcal T_{\mathcal U_{B|CA}^{1'}}=\cup_{j,k}\{l_{j}^{k}\}$, by combining the old tiles $t_{1}^{k}$ with the same row index. In detail, the new tiles $\{l_{j}^{k}\}$ are given by: $l_{0}=t_{0}$, $l_{1}^{k}=t_{1}^{k}\cup t_{8}^{k}$, $l_{2}^{k}=t_{2}^{k}\cup t_{7}^{k}$, $l_{3}^{k}=t_{3}^{k}\cup t_{6}^{k}$ and $l_{4}^{k}=t_{4}^{k}\cup t_{5}^{k}$, where $1\leq k\leq \frac{d-1}{2}$. The new tile structure $\mathcal T_{\mathcal U_{B|CA}^{1'}}=\cup_{j,k}\{l_{j}^{k}\}$ is a U-tile structure. Therefore, the tile structure $\mathcal T_{\mathcal U_{B|CA}^{1}}$ is also a quasi U-tile structure like $\mathcal T_{\mathcal U_{A|BC}^{1}}$. In short, we assert that the UPB $\mathcal U^1$ in $A|BC$ and $B|CA$ bipartitions are SUCPBs by Lemma~\ref{lem:quasi_sucpb}. This completes the proof.
\begin{figure}[htbp]
		\centering
		\includegraphics[width=8cm]{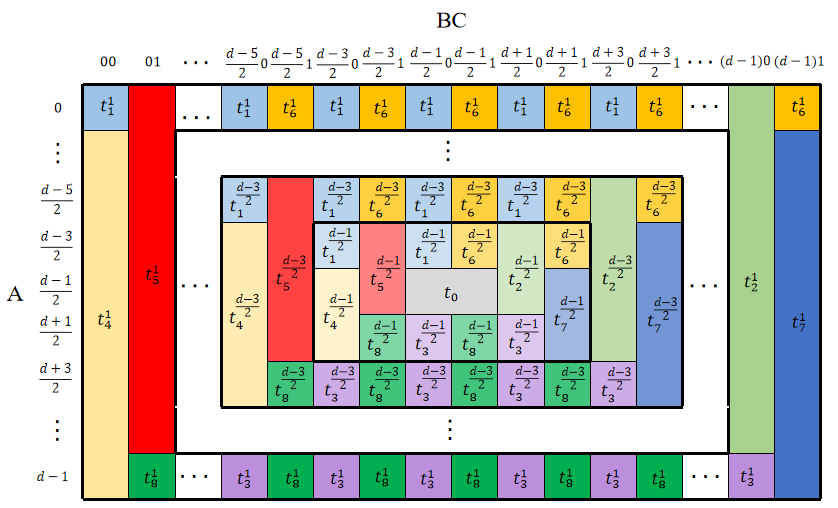}
		\vspace{0.3cm}
		\includegraphics[width=9cm]{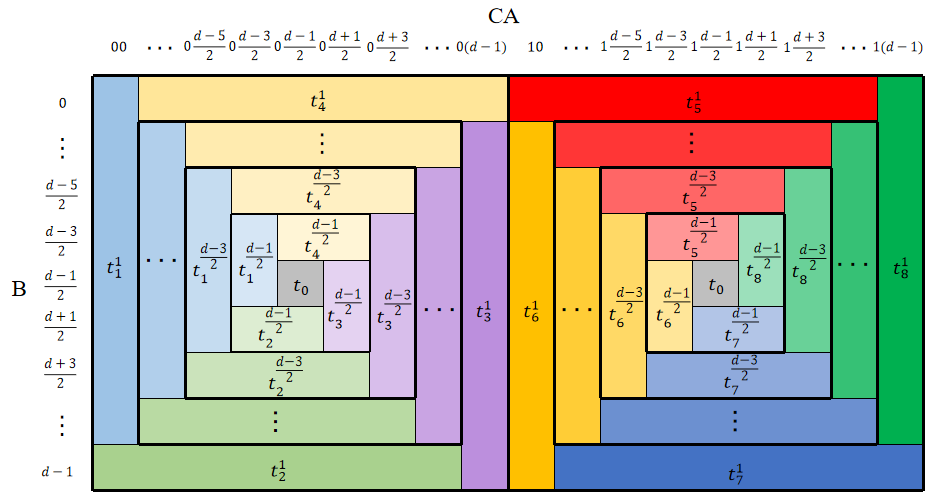}
		\caption{The quasi U-tile structures of $\mathcal T_{\mathcal U_{A|BC}^{1}}$ and $\mathcal T_{\mathcal U_{B|CA}^{1}}$. } \label{fig:d2d}
	\end{figure}
\end{proof}

When $d$ is even, we put forward the following construction of a UPB that has the same structure as $\mathcal U^{1}$ except for the set $\mathcal{A}_{0'}$. Further, we assert that the UPB is an SUCPB in two bipartitions.

\begin{corollary} The set $\cup _{i=1}^{8}\mathcal{A}_{i}^{k}\cup \mathcal{A}_{0'}\cup \{|S\rangle\}$ in $\mathbb{C}^{d}\otimes \mathbb{C}^{d}\otimes \mathbb{C}^{2}$, denoted as $\mathcal U^{1'}$, is a UPB of size $2d^2-4d+8$ for even $d$, $d\geq 4$ and $1\leq k\leq \frac{d-2}{2}$. It is an SUCPB in A$|$BC and B$|$CA bipartitions.
\end{corollary}

\section{UPBs that are SUCPBs in at most one bipartition}\label{sec:sucpbone}

In this section, we construct two types of UPBs with different sizes in $\mathbb{C}^{d}\otimes \mathbb{C}^{2}\otimes \mathbb{C}^{2}$ which are SUCPBs in at most one bipartition for $d\geq 4$. Since any OPS in $\mathbb{C}^{2}\otimes \mathbb{C}^{n}$ can be extended to an OPB, we only need to consider the UPB in A$|$BC bipartition. We denote the basis of $\mathbb{C}^{2}\otimes \mathbb{C}^{2}$ as $|00\rangle \rightarrow |{\bf 0}\rangle$, $|01\rangle \rightarrow |{\bf 1}\rangle$, $|10\rangle \rightarrow |{\bf 2}\rangle$ and $|11\rangle \rightarrow |{\bf 3}\rangle$.

First, we propose the $d\times 2\times 2$ tile structures with 5-tiles, which corresponds to a UPB in $\mathbb{C}^{d}\otimes \mathbb{C}^{2}\otimes \mathbb{C}^{2}$. Note that the UPB is not a strongly uncompletable product basis in any bipartition. It means that we provide a tile structure that is neither a U-tile structure nor a quasi-U-tile structure, such as $\mathcal T_{\mathcal V_{A|BC}^{2}}$ in Fig.~\ref{fig:44}.
\begin{figure}[htbp]
 \centering
 \begin{minipage}[t]{0.48\textwidth}
 \centering
 \includegraphics[width=6.8cm]{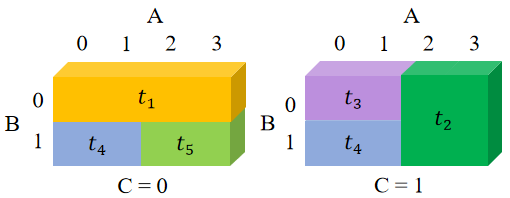}
 \caption{The  $4\times 2 \times 2$ U-tile structure with $5$-tiles.}  \label{fig:422}
 \end{minipage}
 \begin{minipage}[t]{0.48\textwidth}
 \centering
 \includegraphics[width=3.5cm]{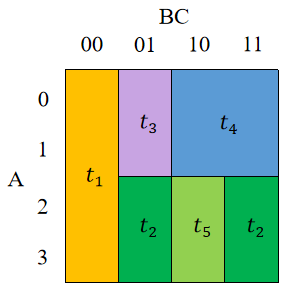}
 \caption{The $4\times 4$ tile structure of $\mathcal T_{\mathcal V_{A|BC}^{2}}$.}  \label{fig:44}
 \end{minipage}
\end{figure}

We give an example of UPB in $\mathbb{C}^{4}\otimes \mathbb{C}^{2}\otimes \mathbb{C}^{2}$ in Fig.~\ref{fig:422}.
\begin{example} Consider the following UPB of size $12$,
\begin{equation} \label{eq:422}
\begin{aligned}
|\phi_{1}\rangle=&|0+1-2-3\rangle_{A}|0\rangle_{B}|0\rangle_{C},~~~~~
|\phi_{2}\rangle=|0-1-2+3\rangle_{A}|0\rangle_{B}|0\rangle_{C},\\
|\phi_{3}\rangle=&|0-1+2-3\rangle_{A}|0\rangle_{B}|0\rangle_{C},~~~~~
|\phi_{4}\rangle=|2+3\rangle_{A}|0-1\rangle_{B}|1\rangle_{C},\\
|\phi_{5}\rangle=&|2-3\rangle_{A}|0+1\rangle_{B}|1\rangle_{C},~~~~~~~~~~
|\phi_{6}\rangle=|2-3\rangle_{A}|0-1\rangle_{B}|1\rangle_{C},\\
|\phi_{7}\rangle=&|0-1\rangle_{A}|0\rangle_{B}|1\rangle_{C},~~~~~~~~~~~~~~~
|\phi_{8}\rangle=|0+1\rangle_{A}|1\rangle_{B}|0-1\rangle_{C},\\
|\phi_{9}\rangle=&|0-1\rangle_{A}|1\rangle_{B}|0+1\rangle_{C},~~~~~~~~~~
|\phi_{10}\rangle=|0-1\rangle_{A}|1\rangle_{B}|0-1\rangle_{C},\\
|\phi_{11}\rangle=&|2-3\rangle_{A}|1\rangle_{B}|0\rangle_{C},~~~~~~~~~~~~~~~
|\phi_{12}\rangle=|0+1+2+3\rangle_{A}|0+1\rangle_{B}|0+1\rangle_{C}.
\end{aligned}
\end{equation}
Denote UPB $\mathcal V^2= \{\cup _{i=1}^{12}|\phi_{i}\rangle\}$ in A$|$BC bipartition as $\mathcal V_{A|BC}^{2}$. The $4\times 4$ tile structure known as $\mathcal T_{\mathcal V_{A|BC}^{2}}$ corresponding to $\mathcal V_{A|BC}^{2}$ is shown in Fig.~\ref{fig:422}. For the tile structure $\mathcal T_{\mathcal V_{A|BC}^{2}}$, there exists only one new partition, that is, the tile $l_0=\cup_{i=1}^{5} t_{i}$. The fact reflects that the impossibility of finding a partition that satisfies the conditions (\ref{item1}) and (\ref{item2}) such that it forms a U-tile structure. In other words, we can find four orthogonal product states in $\mathcal H_{\mathcal V_{A|BC}^{2}}^{\bot}$:
$|\psi_{1}\rangle=|2+3\rangle_{A}(|{\bf 1}\rangle+|{\bf 3}\rangle-2|{\bf 2}\rangle)_{BC}$, $|\psi_{2}\rangle=|0+1\rangle_{A}(|{\bf 2}\rangle+|{\bf 3}\rangle-2|{\bf 1}\rangle)_{BC}$, $|\psi_{3}\rangle=|0+1-2-3\rangle_{A}|{\bf 1+2+3}\rangle_{BC}$ and $|\psi_{4}\rangle=|0+1+2+3\rangle_{A}(|{\bf 1+2+3}\rangle-3|{\bf 0}\rangle)_{BC}$.
Since $\dim (\mathcal H_{\mathcal V_{A|BC}^{2}}^{\bot})=4$, the set $\mathcal V_{A|BC}^{2}$ is a completable product basis. So $\mathcal T_{\mathcal V_{A|BC}^{2}}$ is not a quasi U-tile structure. The UPB $\mathcal V^2$ is not an SUCPB in any bipartition.
\end{example}

Furthermore, generalizing the results of the previous UPB in $\mathbb{C}^{d}\otimes \mathbb{C}^{2}\otimes \mathbb{C}^{2}$:
\begin{equation}
\begin{aligned}
\mathcal{A}_{1}=&\{|\alpha _{i}\rangle_{A}|0\rangle_{B}|0\rangle_{C}\mid i\in \mathbb{Z}_{d}\backslash \{0\}\},~~~~~~~~~
\mathcal{A}_{2}=\{|\xi _{s}\rangle_{A}|\eta _{k}\rangle_{B}|1\rangle_{C}\mid (s,k)\in \mathbb{Z}_{2}\times \mathbb{Z}_{2}\backslash \{(0,0)\}\} ,\\
\mathcal{A}_{3}=&\{|\beta _{j}\rangle_{A}|0\rangle_{B}|1\rangle_{C}\mid j\in \mathbb{Z}_{d-2}\backslash \{0\}\},~~~~~
\mathcal{A}_{4}=\{|\beta _{j}\rangle_{A}|1\rangle_{B}|\eta _{s}\rangle_{C}\mid (j,s)\in \mathbb{Z}_{d-2}\times \mathbb{Z}_{2}\backslash \{(0,0)\}\} ,\\
\mathcal{A}_{5}=&\{|\xi _{s}\rangle_{A}|1\rangle_{B}|0\rangle_{C}\mid s\in \mathbb{Z}_{2}\backslash \{0\}\},
\end{aligned}
\end{equation}
where $|\eta _{s}\rangle _{X}=|0\rangle_{X}+(-1)^{s}|1\rangle_{X}$ and $|\xi _{s}\rangle _{X}=|d-2\rangle_{X}+(-1)^{s}|d-1\rangle_{X}$ for $s\in \mathbb{Z}_{2}$, $|\alpha _{i}\rangle_{X}=\sum \nolimits_{t=0}^{d-1}w_{d}^{it}|t\rangle_{X}$ for $i\in \mathbb Z_{d}$, $|\beta _{j}\rangle_{X}=\sum \nolimits_{t=0}^{d-3}w_{d-2}^{jt}|t\rangle_{X}$ for $j\in \mathbb Z_{d-2}$, $X\in \{A, B, C\}$, we have the following conclusion.

\begin{theorem} The set $\cup _{i=1}^{5}\mathcal{A}_{i}\cup\{|S\rangle\}$ in $\mathbb{C}^{d}\otimes \mathbb{C}^{2}\otimes \mathbb{C}^{2}$, denoted as $\mathcal U^2$, is a UPB of size $4d-4$ that is not an SUCPB in any bipartition for $d\geq 4$.
\end{theorem}

\begin{proof} Similarly, we only need to consider the UPB $\mathcal U^{2}$ in A$|$BC bipartition, denoted as $\mathcal U_{A|BC}^{2}$. In $\mathcal H_{\mathcal U_{A|BC}^{2}}^{\bot}$, the following four orthogonal product states can be obtained: $|\psi_{1}\rangle=|\xi _{0}\rangle_{A}(|{\bf 1}\rangle+|{\bf 3}\rangle-2|{\bf 2}\rangle)_{BC}$, $|\psi_{2}\rangle=|\beta _{0}\rangle_{A}(|{\bf 2}\rangle+|{\bf 3}\rangle-2|{\bf 1}\rangle)_{BC}$, $|\psi_{3}\rangle=[2|\beta _{0}\rangle-(d-2)|\xi _{0}\rangle]_{A}|{\bf 1+2+3}\rangle_{BC}$ and $|\psi_{4}\rangle=|\alpha _{0}\rangle_{A}(|{\bf 1+2+3}\rangle-3|{\bf 0}\rangle)_{BC}$.
Since $\dim (\mathcal H_{\mathcal U_{A|BC}^{2}}^{\bot})=4$, the set $\mathcal U_{A|BC}^{2}$ can be extended to an OPB. Therefore, the UPB $\mathcal U^2$ is not an SUCPB in any bipartition. This completes the proof.
\end{proof}

Next, we present a new decomposition of the $d\times 2\times 2$ tile structure with $8$-tiles, which corresponds to a UPB with different size in $\mathbb{C}^{d}\otimes \mathbb{C}^{2}\otimes \mathbb{C}^{2}$. The difference is that the UPB in $A|BC$ bipartition is an SUCPB. First of all, we construct a UPB in $\mathbb{C}^{4}\otimes \mathbb{C}^{2}\otimes \mathbb{C}^{2}$ in Fig.~\ref{fig:422UPB}.
\begin{figure}[htbp]
\centering
\begin{minipage}[t]{0.48\textwidth}
\centering
\includegraphics[width=7cm]{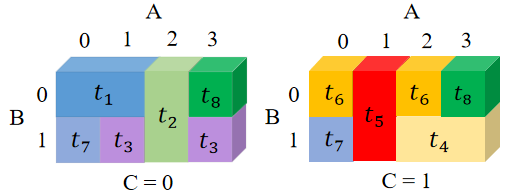}
\caption{The $4\times 2 \times 2$ U-tile structure with $8$-tiles.}  \label{fig:422UPB}
\end{minipage}
\begin{minipage}[t]{0.48\textwidth}
\centering
\includegraphics[width=4cm]{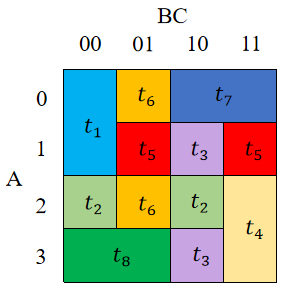}
\caption{The $4\times 4$ U-tile structure of $\mathcal T_{\mathcal V_{A|BC}^{3}}$.}  \label{fig:44UPB}
\end{minipage}
\end{figure}

\begin{proposition} In $\mathbb{C}^{4}\otimes \mathbb{C}^{2}\otimes \mathbb{C}^{2}$, the UPB of size $9$ given below is an SUCPB in A$|$BC bipartition:
\begin{equation}\label{eq:422UPB}
\begin{aligned}
|\phi_{1}\rangle=&|0-1\rangle_{A}|0\rangle_{B}|0\rangle_{C},~~~~~ |\phi_{5}\rangle=|1\rangle_{A}|0-1\rangle_{B}|1\rangle_{C},\\
|\phi_{2}\rangle=&|2\rangle_{A}|0-1\rangle_{B}|0\rangle_{C},~~~~~ |\phi_{6}\rangle=|0-2\rangle_{A}|0\rangle_{B}|1\rangle_{C},\\
|\phi_{3}\rangle=&|1-3\rangle_{A}|1\rangle_{B}|0\rangle_{C},~~~~~ |\phi_{7}\rangle=|0\rangle_{A}|1\rangle_{B}|0-1\rangle_{C},\\
|\phi_{4}\rangle=&|2-3\rangle_{A}|1\rangle_{B}|1\rangle_{C},~~~~~ |\phi_{8}\rangle=|3\rangle_{A}|0\rangle_{B}|0-1\rangle_{C},\\
|\phi_{9}\rangle=&|0+1+2+3\rangle_{A}|0+1\rangle_{B}|0+1\rangle_{C}. \\
\end{aligned}
\end{equation}
\end{proposition}

\begin{proof} Denote the UPB as $\mathcal V^3=\{\cup _{i=1}^{9}|\phi_{i}\rangle\}$, which is a completable orthogonal product basis in B$|$CA or C$|$AB bipartitions. We only need to consider the UPB $\mathcal V^3$ in A$|$BC bipartition, denoted as $\mathcal V_{A|BC}^{3}$. By Definition the set $\mathcal V_{A|BC}^{3}$ corresponds to the tile structure with $8$-tiles, denoted as $\cT_{\mathcal V_{A|BC}^{3}}$, which is a U-tile structure in Fig.~\ref{fig:44UPB}. Thus, the set $\mathcal V_{A|BC}^{3}$ is a UPB and also an SUCPB.
\end{proof}

Subsequently, we propose the general construction of UPB in $\mathbb{C}^{d}\otimes \mathbb{C}^{2}\otimes \mathbb{C}^{2}$:
\begin{equation}
\begin{aligned}
\mathcal{A}_{1}&=\{|\eta _{s}\rangle_{A}|0\rangle_{B}|0\rangle_{C}\mid s\in \mathbb{Z}_{2}\backslash \{0\}\},~~~~~~~~
\mathcal{A}_{2}=\{|\alpha _{i}\rangle_{A}|\eta _{s}\rangle_{B}|0\rangle_{C}\mid (i,s)\in \mathbb{Z}_{d-3}\times \mathbb{Z}_{2}\backslash \{(0,0)\}\},\\
\mathcal{A}_{3}&=\{|\varepsilon _{s}\rangle_{A}|1\rangle_{B}|0\rangle_{C}\mid s\in \mathbb{Z}_{2}\backslash \{0\}\},~~~~~~~~
\mathcal{A}_{4}=\{|\xi _{s}\rangle_{A}|1\rangle_{B}|1\rangle_{C}\}\mid s\in \mathbb{Z}_{2}\backslash \{0\}\},\\
\mathcal{A}_{5}&=\{|\beta _{j}\rangle_{A}|\eta _{s}\rangle_{V}|1\rangle_{C}\}\mid (j,s)\in \mathbb{Z}_{d-3}\times \mathbb{Z}_{2}\backslash \{(0,0)\}\},~~~~~~~~
\mathcal{A}_{6}=\{|\zeta _{s}\rangle_{A}|0\rangle_{B}|1\rangle_{C}\mid s\in \mathbb{Z}_{2}\backslash \{0\}\},\\
\mathcal{A}_{7}&=\{|0\rangle_{A}|1\rangle_{B}|\eta _{s}\rangle_{C}\mid s\in \mathbb{Z}_{2}\backslash \{0\}\},~~~~~~~~
\mathcal{A}_{8}=\{|d-1\rangle_{A}|0\rangle_{B}|\eta _{s}\rangle_{C}\mid s\in \mathbb{Z}_{2}\backslash \{0\}\},\\
\end{aligned}
\end{equation}
where $|\eta _{s}\rangle_{X}=|0\rangle_{X}+(-1)^{s}|1\rangle_{X}$,
$|\xi _{s}\rangle_{X}=|d-2\rangle_{X}+(-1)^{s}|d-1\rangle_{X}$, $|\zeta _{s}\rangle_{X}=|0\rangle_{X}+(-1)^{s}|d-2\rangle_{X}$ and $|\varepsilon_{s}\rangle_{X}=|1\rangle_{X}+(-1)^{s}|d-1\rangle_{X}$ for $s\in \mathbb{Z}_{2}$, $|\alpha _{i}\rangle_{X}=\sum_{t\in \mathbb Z_{d-3}}w_{d-3}^{it}|t+2\rangle_{X}$ for $i\in \mathbb Z_{d-3}$, $|\beta _{j}\rangle_{X}=\sum_{t\in \mathbb Z_{d-3}}w_{d-3}^{jt}|t+1\rangle_{X}$ for $j\in \mathbb Z_{d-3}$, $X\in \{A, B, C\}$.
Therefore, we obtain the following result:

\begin{theorem}
In $\mathbb{C}^{d}\otimes \mathbb{C}^{2}\otimes \mathbb{C}^{2}$, the set $\cup _{i=1}^{8}\mathcal{A}_{i}\cup \{|S\rangle\}$ denoted as $\mathcal U^3$ is a UPB of size $4d-7$, which is an SUCPB in A$|$BC bipartition for $d\geq 4$.
\end{theorem}

\section{UPBs that are SUCPBs in every bipartition}\label{sec:sucpbevery}

We employ a UPB formed by the stopper state and the``vertical tile'' states $|V_{ki}^{-}\rangle$, ``horizontal tile'' states $|H_{ki}^{-}\rangle$, and ``crossed tile'' states $|C_{ki}^{-}\rangle$ in $\mathbb{C}^{6}\otimes \mathbb{C}^{6}\otimes \mathbb{C}^{6}$ from Table~\ref{tab:666} to construct SUCPBs, where $k,i\in \mathbb Z_{6}$.
\begin{table}[htbp]
    \newcommand{\tabincell}[2]{\begin{tabular}{@{}#1@{}}#2\end{tabular}}
    \centering
    \caption{\label{tab:666}The orthogonal product states of size $108$ in $\mathbb{C}^{6}\otimes \mathbb{C}^{6}\otimes \mathbb{C}^{6}$.}
    \begin{tabular}{cccccc}
        \hline
        \hline
        \specialrule{0em}{1.5pt}{1.5pt}
        ``Vertical tile''~&``Horizontal tile'' &~``Crossed tile'' ~~&``Vertical tile'' ~&``Horizontal tile'' &~``Crossed tile'' \\
        \specialrule{0em}{1.5pt}{1.5pt}
        \hline
        \specialrule{0em}{1.5pt}{1.5pt}
        \tabincell{c}{$|V_{00}^{-}\rangle=|0\rangle|2-3\rangle|0\rangle$}&
        \tabincell{c}{$|H_{00}^{-}\rangle=|1-2\rangle|0\rangle|0\rangle$}&
        \tabincell{c}{$|C_{00}^{-}\rangle=|0\rangle|0\rangle|5-0\rangle$}&
        \tabincell{c}{$|V_{03}^{-}\rangle=|0\rangle|5-0\rangle|3\rangle$}&
        \tabincell{c}{$|H_{03}^{-}\rangle=|4-5\rangle|0\rangle|3\rangle$}&
        \tabincell{c}{$|C_{03}^{-}\rangle=|0\rangle|3\rangle|2-3\rangle$} \\
        \specialrule{0em}{1.5pt}{1.5pt}
        \tabincell{c}{$|V_{10}^{-}\rangle=|1\rangle|3-4\rangle|0\rangle$}&
        \tabincell{c}{$|H_{10}^{-}\rangle=|2-3\rangle|1\rangle|0\rangle$}&
        \tabincell{c}{$|C_{10}^{-}\rangle=|1\rangle|0\rangle|4-5\rangle$}&
        \tabincell{c}{$|V_{13}^{-}\rangle=|1\rangle|0-1\rangle|3\rangle$}&
        \tabincell{c}{$|H_{13}^{-}\rangle=|5-0\rangle|1\rangle|3\rangle$}&
        \tabincell{c}{$|C_{13}^{-}\rangle=|1\rangle|3\rangle|1-2\rangle$} \\
        \specialrule{0em}{1.5pt}{1.5pt}
        \tabincell{c}{$|V_{20}^{-}\rangle=|2\rangle|4-5\rangle|0\rangle$}&
        \tabincell{c}{$|H_{20}^{-}\rangle=|3-4\rangle|2\rangle|0\rangle$}&
        \tabincell{c}{$|C_{20}^{-}\rangle=|2\rangle|0\rangle|3-4\rangle$}&
        \tabincell{c}{$|V_{23}^{-}\rangle=|2\rangle|1-2\rangle|3\rangle$}&
        \tabincell{c}{$|H_{23}^{-}\rangle=|0-1\rangle|2\rangle|3\rangle$}&
        \tabincell{c}{$|C_{23}^{-}\rangle=|2\rangle|3\rangle|0-1\rangle$} \\
        \specialrule{0em}{1.5pt}{1.5pt}
        \tabincell{c}{$|V_{30}^{-}\rangle=|3\rangle|5-0\rangle|0\rangle$}&
        \tabincell{c}{$|H_{30}^{-}\rangle=|4-5\rangle|3\rangle|0\rangle$}&
        \tabincell{c}{$|C_{30}^{-}\rangle=|3\rangle|0\rangle|2-3\rangle$}&
        \tabincell{c}{$|V_{33}^{-}\rangle=|3\rangle|2-3\rangle|3\rangle$}&
        \tabincell{c}{$|H_{33}^{-}\rangle=|1-2\rangle|3\rangle|3\rangle$}&
        \tabincell{c}{$|C_{33}^{-}\rangle=|3\rangle|3\rangle|5-0\rangle$} \\
        \specialrule{0em}{1.5pt}{1.5pt}
        \tabincell{c}{$|V_{40}^{-}\rangle=|4\rangle|0-1\rangle|0\rangle$}&
        \tabincell{c}{$|H_{40}^{-}\rangle=|5-0\rangle|4\rangle|0\rangle$}&
        \tabincell{c}{$|C_{40}^{-}\rangle=|4\rangle|0\rangle|1-2\rangle$}&
        \tabincell{c}{$|V_{43}^{-}\rangle=|4\rangle|3-4\rangle|3\rangle$}&
        \tabincell{c}{$|H_{43}^{-}\rangle=|2-3\rangle|4\rangle|3\rangle$}&
        \tabincell{c}{$|C_{43}^{-}\rangle=|4\rangle|3\rangle|4-5\rangle$} \\
        \specialrule{0em}{1.5pt}{1.5pt}
        \tabincell{c}{$|V_{50}^{-}\rangle=|5\rangle|1-2\rangle|0\rangle$}&
        \tabincell{c}{$|H_{50}^{-}\rangle=|0-1\rangle|5\rangle|0\rangle$}&
        \tabincell{c}{$|C_{50}^{-}\rangle=|5\rangle|0\rangle|0-1\rangle$}&
        \tabincell{c}{$|V_{53}^{-}\rangle=|5\rangle|4-5\rangle|3\rangle$}&
        \tabincell{c}{$|H_{53}^{-}\rangle=|3-4\rangle|5\rangle|3\rangle$}&
        \tabincell{c}{$|C_{53}^{-}\rangle=|5\rangle|3\rangle|3-4\rangle$} \\
        \hline
        \specialrule{0em}{3.5pt}{3.5pt}
        \tabincell{c}{$|V_{01}^{-}\rangle=|0\rangle|3-4\rangle|1\rangle$}&
        \tabincell{c}{$|H_{01}^{-}\rangle=|0-1\rangle|0\rangle|1\rangle$}&
        \tabincell{c}{$|C_{01}^{-}\rangle=|0\rangle|1\rangle|0-1\rangle$}&
        \tabincell{c}{$|V_{04}^{-}\rangle=|0\rangle|0-1\rangle|4\rangle$}&
        \tabincell{c}{$|H_{04}^{-}\rangle=|3-4\rangle|0\rangle|4\rangle$}&
        \tabincell{c}{$|C_{04}^{-}\rangle=|0\rangle|4\rangle|3-4\rangle$} \\
        \specialrule{0em}{1.5pt}{1.5pt}
        \tabincell{c}{$|V_{11}^{-}\rangle=|1\rangle|4-5\rangle|1\rangle$}&
        \tabincell{c}{$|H_{11}^{-}\rangle=|1-2\rangle|1\rangle|1\rangle$}&
        \tabincell{c}{$|C_{11}^{-}\rangle=|1\rangle|1\rangle|5-0\rangle$}&
        \tabincell{c}{$|V_{14}^{-}\rangle=|1\rangle|1-2\rangle|4\rangle$}&
        \tabincell{c}{$|H_{14}^{-}\rangle=|4-5\rangle|1\rangle|4\rangle$}&
        \tabincell{c}{$|C_{14}^{-}\rangle=|1\rangle|4\rangle|2-3\rangle$} \\
        \specialrule{0em}{1.5pt}{1.5pt}
        \tabincell{c}{$|V_{21}^{-}\rangle=|2\rangle|5-0\rangle|1\rangle$}&
        \tabincell{c}{$|H_{21}^{-}\rangle=|2-3\rangle|2\rangle|1\rangle$}&
        \tabincell{c}{$|C_{21}^{-}\rangle=|2\rangle|1\rangle|4-5\rangle$}&
        \tabincell{c}{$|V_{24}^{-}\rangle=|2\rangle|2-3\rangle|4\rangle$}&
        \tabincell{c}{$|H_{24}^{-}\rangle=|5-0\rangle|2\rangle|4\rangle$}&
        \tabincell{c}{$|C_{24}^{-}\rangle=|2\rangle|4\rangle|1-2\rangle$} \\
        \specialrule{0em}{1.5pt}{1.5pt}
        \tabincell{c}{$|V_{31}^{-}\rangle=|3\rangle|0-1\rangle|1\rangle$}&
        \tabincell{c}{$|H_{31}^{-}\rangle=|3-4\rangle|3\rangle|1\rangle$}&
        \tabincell{c}{$|C_{31}^{-}\rangle=|3\rangle|1\rangle|3-4\rangle$}&
        \tabincell{c}{$|V_{34}^{-}\rangle=|3\rangle|3-4\rangle|4\rangle$}&
        \tabincell{c}{$|H_{34}^{-}\rangle=|0-1\rangle|3\rangle|4\rangle$}&
        \tabincell{c}{$|C_{34}^{-}\rangle=|3\rangle|4\rangle|0-1\rangle$} \\
        \specialrule{0em}{1.5pt}{1.5pt}
        \tabincell{c}{$|V_{41}^{-}\rangle=|4\rangle|1-2\rangle|1\rangle$}&
        \tabincell{c}{$|H_{41}^{-}\rangle=|4-5\rangle|4\rangle|1\rangle$}&
        \tabincell{c}{$|C_{41}^{-}\rangle=|4\rangle|1\rangle|2-3\rangle$}&
        \tabincell{c}{$|V_{44}^{-}\rangle=|4\rangle|4-5\rangle|4\rangle$}&
        \tabincell{c}{$|H_{44}^{-}\rangle=|1-2\rangle|4\rangle|4\rangle$}&
        \tabincell{c}{$|C_{44}^{-}\rangle=|4\rangle|4\rangle|5-0\rangle$} \\
        \specialrule{0em}{1.5pt}{1.5pt}
        \tabincell{c}{$|V_{51}^{-}\rangle=|5\rangle|2-3\rangle|1\rangle$}&
        \tabincell{c}{$|H_{51}^{-}\rangle=|5-0\rangle|5\rangle|1\rangle$}&
        \tabincell{c}{$|C_{51}^{-}\rangle=|5\rangle|1\rangle|1-2\rangle$}&
        \tabincell{c}{$|V_{54}^{-}\rangle=|5\rangle|5-0\rangle|4\rangle$}&
        \tabincell{c}{$|H_{54}^{-}\rangle=|2-3\rangle|5\rangle|4\rangle$}&
        \tabincell{c}{$|C_{54}^{-}\rangle=|5\rangle|4\rangle|4-5\rangle$} \\
        \hline
        \specialrule{0em}{3.5pt}{3.5pt}
        \tabincell{c}{$|V_{02}^{-}\rangle=|0\rangle|4-5\rangle|2\rangle$}&
        \tabincell{c}{$|H_{02}^{-}\rangle=|5-0\rangle|0\rangle|2\rangle$}&
        \tabincell{c}{$|C_{02}^{-}\rangle=|0\rangle|2\rangle|1-2\rangle$}&
        \tabincell{c}{$|V_{05}^{-}\rangle=|0\rangle|1-2\rangle|5\rangle$}&
        \tabincell{c}{$|H_{05}^{-}\rangle=|2-3\rangle|0\rangle|5\rangle$}&
        \tabincell{c}{$|C_{05}^{-}\rangle=|0\rangle|5\rangle|4-5\rangle$} \\
        \specialrule{0em}{1.5pt}{1.5pt}
        \tabincell{c}{$|V_{12}^{-}\rangle=|1\rangle|5-0\rangle|2\rangle$}&
        \tabincell{c}{$|H_{12}^{-}\rangle=|0-1\rangle|1\rangle|2\rangle$}&
        \tabincell{c}{$|C_{12}^{-}\rangle=|1\rangle|2\rangle|0-1\rangle$}&
        \tabincell{c}{$|V_{15}^{-}\rangle=|1\rangle|2-3\rangle|5\rangle$}&
        \tabincell{c}{$|H_{15}^{-}\rangle=|3-4\rangle|1\rangle|5\rangle$}&
        \tabincell{c}{$|C_{15}^{-}\rangle=|1\rangle|5\rangle|3-4\rangle$} \\
        \specialrule{0em}{1.5pt}{1.5pt}
        \tabincell{c}{$|V_{22}^{-}\rangle=|2\rangle|0-1\rangle|2\rangle$}&
        \tabincell{c}{$|H_{22}^{-}\rangle=|1-2\rangle|2\rangle|2\rangle$}&
        \tabincell{c}{$|C_{22}^{-}\rangle=|2\rangle|2\rangle|5-0\rangle$}&
        \tabincell{c}{$|V_{25}^{-}\rangle=|2\rangle|3-4\rangle|5\rangle$}&
        \tabincell{c}{$|H_{25}^{-}\rangle=|4-5\rangle|2\rangle|5\rangle$}&
        \tabincell{c}{$|C_{25}^{-}\rangle=|2\rangle|5\rangle|2-3\rangle$} \\
        \specialrule{0em}{1.5pt}{1.5pt}
        \tabincell{c}{$|V_{32}^{-}\rangle=|3\rangle|1-2\rangle|2\rangle$}&
        \tabincell{c}{$|H_{32}^{-}\rangle=|2-3\rangle|3\rangle|2\rangle$}&
        \tabincell{c}{$|C_{32}^{-}\rangle=|3\rangle|2\rangle|4-5\rangle$}&
        \tabincell{c}{$|V_{35}^{-}\rangle=|3\rangle|4-5\rangle|5\rangle$}&
        \tabincell{c}{$|H_{35}^{-}\rangle=|5-0\rangle|3\rangle|5\rangle$}&
        \tabincell{c}{$|C_{35}^{-}\rangle=|3\rangle|5\rangle|1-2\rangle$} \\
        \specialrule{0em}{1.5pt}{1.5pt}
        \tabincell{c}{$|V_{42}^{-}\rangle=|4\rangle|2-3\rangle|2\rangle$}&
        \tabincell{c}{$|H_{42}^{-}\rangle=|3-4\rangle|4\rangle|2\rangle$}&
        \tabincell{c}{$|C_{42}^{-}\rangle=|4\rangle|2\rangle|3-4\rangle$}&
        \tabincell{c}{$|V_{45}^{-}\rangle=|4\rangle|5-0\rangle|5\rangle$}&
        \tabincell{c}{$|H_{45}^{-}\rangle=|0-1\rangle|4\rangle|5\rangle$}&
        \tabincell{c}{$|C_{45}^{-}\rangle=|4\rangle|5\rangle|0-1\rangle$} \\
        \specialrule{0em}{1.5pt}{1.5pt}
        \tabincell{c}{$|V_{52}^{-}\rangle=|5\rangle|3-4\rangle|2\rangle$}&
        \tabincell{c}{$|H_{52}^{-}\rangle=|4-5\rangle|5\rangle|2\rangle$}&
        \tabincell{c}{$|C_{52}^{-}\rangle=|5\rangle|2\rangle|2-3\rangle$}&
        \tabincell{c}{$|V_{55}^{-}\rangle=|5\rangle|0-1\rangle|5\rangle$}&
        \tabincell{c}{$|H_{55}^{-}\rangle=|1-2\rangle|5\rangle|5\rangle$}&
        \tabincell{c}{$|C_{55}^{-}\rangle=|5\rangle|5\rangle|5-0\rangle$} \\
        \specialrule{0em}{1.5pt}{1.5pt}
        \hline
        \hline
    \end{tabular}
\end{table}

\begin{proposition}
In $\mathbb{C}^{6}\otimes \mathbb{C}^{6}\otimes \mathbb{C}^{6}$, the UPB $\mathcal V^4$ of size $109$ given by Eq. \eqref{eq:666} is an SUCPB in every bipartition:
\begin{equation} \label{eq:666}
\begin{aligned}
\mathcal V^4=&\cup_{k,i=0}^{5}\{|V_{ki}^{-}\rangle\cup |H_{ki}^{-}\rangle\cup |C_{ki}^{-}\rangle\}\cup \{|S\rangle\}.
\end{aligned}
\end{equation}
\end{proposition}

\begin{proof} The UPB $\mathcal V^4$ in every bipartition is denoted by the sets $\mathcal V_{A|BC}^4$, $\mathcal V_{B|CA}^4$ and $\mathcal V_{C|AB}^4$. The tiles $x_{ki}$ correspond to states $|X_{ki}^{-}\rangle$ of $\mathcal V^4$, where $x\in \{v, h, c\}$, $X\in \{V, H, C\}$ and $k,i\in \mathbb Z_6$. According to Lemma~\ref{lem:quasi_sucpb}, $\mathcal V_{A|BC}^4$, $\mathcal V_{B|CA}^4$ and $\mathcal V_{C|AB}^4$ are SUCPBs if the tile structures, corresponding to the three sets $\mathcal T_{\mathcal V_{A|BC}^{4}}$, $\mathcal T_{\mathcal V_{B|CA}^{4}}$ and $\mathcal T_{\mathcal V_{C|AB}^{4}}$, are quasi U-tile structures in Fig.~\ref{fig:636}.
\begin{figure}[htbp]
 \centering
 \includegraphics[width=17cm]{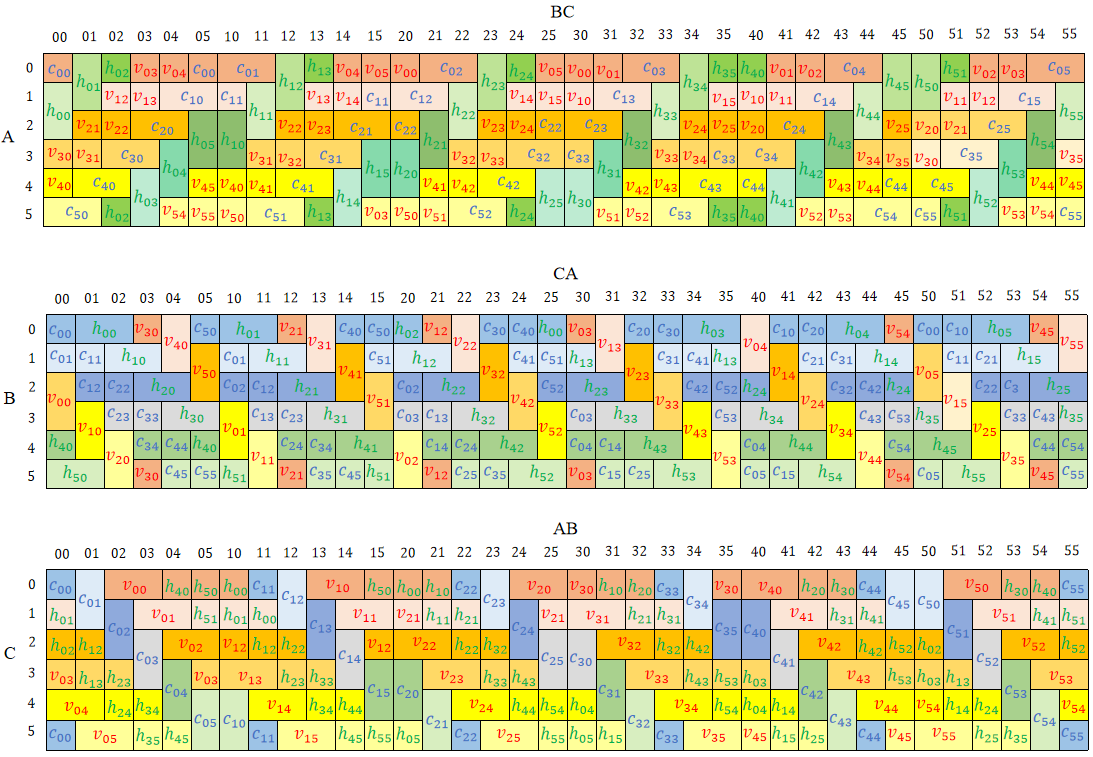}
 \caption{The quasi U-tile structures of $\mathcal T_{\mathcal V_{A|BC}^{4}}$, $\mathcal T_{\mathcal V_{B|CA}^{4}}$ and $\mathcal T_{\mathcal V_{C|AB}^{4}}$.}  \label{fig:636}
\end{figure}

For the tile structure $\mathcal T_{\mathcal V_{A|BC}^{4}}$ in Fig.~\ref{fig:636}, there exists a new partition that reorganizes the tiles $x_{ki}$ to form the new tiles $l_j$, i.e., the tiles $l_k=\cup_{i=0}^{5} \{v_{ki}\cup c_{ki}\}$ and $l_{6+k}=\cup_{i=0}^{5} \{h_{i(i+k~mod~6)}\}$, where $k\in \mathbb Z_6$ and $j\in \mathbb Z_{12}$. Since the new tiles $\{l_j\}_{j=0}^{11}$ satisfy the conditions (\ref{item1}) and (\ref{item2}), the tile structure $\cup _{j=0}^{11}l_j$ is a U-tile structure. $\mathcal T_{\mathcal V_{A|BC}^{4}}=\cup_{k,i=0}^{5}\{v_{ki}\cup h_{ki}\cup c_{ki}\}$ is a quasi U-tile structure. It means that the set $\mathcal V_{A|BC}^4$ is an SUCPB.

Consider the tile structure $\mathcal T_{\mathcal V_{B|CA}^{4}}$ in Fig.~\ref{fig:636}, we put forward a partition of the new tiles $l_j$, that is, the tiles $l_k=\cup_{i=0}^{5} \{h_{ki}\cup c_{ik}\}$ and $l_{6+k}=\cup_{i=0}^{5} \{v_{i(6-k-i~mod~6)}\}$, where $k\in
\mathbb Z_6$ and $j\in \mathbb Z_{12}$. The new tile structure $\cup_{j=0}^{11}\{l_j\}$ is a U-tile structure. According Lemma~\ref{lem:quasi_sucpb}, the tile structure $\mathcal T_{\mathcal V_{B|CA}^{4}}$ is a quasi U-tile structure and the set $\mathcal V_{B|CA}^{4}$ is an SUCPB.

Lastly, we consider the tile structure $\mathcal T_{\mathcal V_{C|AB}^{4}}$ in Fig.~\ref{fig:636}. We have the new partition consisting of the tiles $l_k=\cup_{i=0}^{5} \{v_{ik}\cup h_{ik}\}$ and $l_{6+k}=\cup_{i=0}^{5} \{c_{i(i+k~mod~6)}\}$, where $k\in \mathbb Z_6$. The new tile structure $\cup_{j=0}^{11}\{l_j\}$ is a U-tile structure. Hence, $\mathcal T_{\mathcal V_{C|AB}^{4}}$ is a quasi U-tile structure and the set $\mathcal V_{C|BA}^4$ is an SUCPB. Therefore, the UPB $\mathcal V^4$ is an SUCPB in every bipartition.
\end{proof}

In Ref.~\ref{Shi2022NJP}, Shi {\it et al}. proposed a UPB of size $200$ in $\mathbb{C}^{6}\otimes \mathbb{C}^{6}\otimes \mathbb{C}^{6}$ which is an SUCPB in every bipartition. Here, we construct a UPB of size $109$, fewer than that of Shi {\it et al}., which is also an SUCPB in every bipartition. More importantly, the UPB $\mathcal V^4$ consists of 108 $1\times 1\times 2$ tiles, illustrating that it is the UPB with the smallest cardinality by using the U-tile structure in $\mathbb{C}^{6}\otimes \mathbb{C}^{6}\otimes \mathbb{C}^{6}$.

Next, we provide the general forms of ``vertical tile'' states $|V_{ki}^{m}\rangle$, ``horizontal tile'' states $|H_{ki}^{n}\rangle$ and ``crossed tile'' states $|C_{ki}^{s}\rangle$ in $\mathbb{C}^{d}\otimes \mathbb{C}^{d}\otimes \mathbb{C}^{d}$, where $d\geq 6$, $k,i\in \mathbb{Z}_{d}$. We denote
\begin{equation}\label{eq:ddd}
\begin{aligned}
|V_{ki}^{m}\rangle&=|k\rangle_{A}|\alpha _{m}\rangle_{B}|i\rangle_{C}=|k\rangle_{A}\left(\sum \limits_{j=0}^{\lfloor \frac{d}{2}\rfloor-2}w_{\lfloor \frac{d}{2}\rfloor-1}^{jm}|j+k+2+i~mod~d\rangle\right)_{B}|i\rangle_{C},~~~m\in \mathbb Z_{\lfloor \frac{d}{2}\rfloor-1}\\
|H_{ki}^{n}\rangle&=|\beta _{n}\rangle_{A}|k\rangle_{B}|i\rangle_{C}=\left(\sum \limits_{j=0}^{\lceil \frac{d}{2}\rceil-2}w_{\lceil \frac{d}{2}\rceil-1}^{jn}|j+k+1-i~mod~d\rangle\right)_{A}|k\rangle_{B}|i\rangle_{C},~~~n\in \mathbb Z_{\lceil \frac{d}{2}\rceil-1}\\
|C_{ki}^{s}\rangle&=|k\rangle_{A}|i\rangle_{B}|\eta _{s}\rangle_{C}=|k\rangle_{A}|i\rangle_{B}(|i-k-1~mod~d\rangle+(-1)^{s}|i-k~mod~d\rangle)_{C},~~~ s\in \mathbb{Z}_{2}.
\end{aligned}
\end{equation}

\begin{theorem}\label{th:ddd}  In $\mathbb{C}^{d}\otimes \mathbb{C}^{d}\otimes \mathbb{C}^{d}$, the set $\cup_{k,i}(\cup_{m,n,s}\{|V_{ki}^{m}\rangle\cup |H_{ki}^{n}\rangle\cup |C_{ki}^{s}\rangle\}\backslash \{|V_{ki}^{0}\rangle\cup |H_{ki}^{0}\rangle\cup |C_{ki}^{0}\rangle\}) \cup \{|S\rangle\}$, denoted as $\mathcal U^4$, is a UPB of size $d^3-3d^2+1$ and an SUCPB in every bipartition for $d\geq 6$.
\end{theorem}

\begin{proof} Denoted $\mathcal U_{A|BC}^4$, $\mathcal U_{B|CA}^4$ and $\mathcal U_{C|AB}^4$
the UPB $\mathcal U^4$ in every bipartition. The tiles $x_{ki}^{y}$ correspond to states $|X_{ki}^{y}\rangle$ of $\mathcal U^4$, where $x\in \{v, h, c\}$, $y\in \{m, n, s\}$ and $X\in \{V, H, C\}$ and $k,i\in \mathbb Z_6$. If the tile structures $\mathcal T_{\mathcal U_{A|BC}^{4}}$, $\mathcal T_{\mathcal U_{B|CA}^{4}}$ and $\mathcal T_{\mathcal U_{C|AB}^{4}}$ corresponding to the three sets are quasi U-tile structures, the sets $\mathcal U_{A|BC}^4$, $\mathcal U_{B|CA}^4$ and $\mathcal U_{C|AB}^4$ are SUCPBs by Lemma~\ref{lem:quasi_sucpb}.

Since the tile structure $\mathcal T_{\mathcal U_{A|BC}^{4}}$ is similar to $\mathcal T_{\mathcal V_{A|BC}^{4}}$, there is a new partition consisting of the tiles $l_j$, i.e., the tiles $l_k=\cup_{i=0}^{d-1} \{v_{ki}\cup c_{ki}\}$ and $l_{d+k}=\cup_{i=0}^{d-1} \{h_{i(i+k~mod~d)}\}$. Similarly, for the tile structure $\mathcal T_{\mathcal U_{B|CA}^{4}}$, the new partition is formed by $l_j$, that is, the tiles $l_k=\cup_{i=0}^{d-1}~\{h_{ki}\cup c_{ik}\}$ and $l_{d+k}=\cup_{i=0}^{d-1}~\{v_{i(d-k-i~mod~d)}\}$. For the tile structure $\mathcal T_{\mathcal U_{C|AB}^{4}}$, we get the new partition consisting of the tiles $l_k=\cup_{i=0}^{d-1} \{v_{ik}\cup h_{ik}\}$ and $l_{d+k}=\cup_{i=0}^{d-1} \{c_{i(i+k~mod~d)}\}$, where $k\in \mathbb Z_d$. $j\in \mathbb Z_{2d}$. In every bipartition, the new tiles $\{l_j\}_{j=0}^{2d-1}$ satisfy the conditions (\ref{item1}) and (\ref{item2}). They form a U-tile structure. Then, the tile structures $\mathcal T_{\mathcal U_{A|BC}^{4}}$, $\mathcal T_{\mathcal U_{B|CA}^{4}}$ and $\mathcal T_{\mathcal U_{C|AB}^{4}}$ are quasi U-tile structures and the corresponding sets $\mathcal U_{A|BC}^{4}$, $\mathcal U_{B|CA}^{4}$ and $\mathcal U_{C|AB}^{4}$ are SUCPBs. In short, the UPB $\mathcal U^4$ is an SUCPB in every bipartition.
\end{proof}

\section{Conclusion}
Due to the close connection of UPBs with bound entangled states and quantum nonlocality without entanglement, it is of great significance to discuss the relationship between UPBs and SUCPBs in every bipartition. To show that an orthogonal product set is an SUCPB, we defined a tile structure called a quasi U-tile structure that perfectly corresponds to an SUCPB. We have generalized the TILES UPB given by Bennett {\it et al}. to construct UPBs of sizes $2d^2-4d+4$ and $2d^2-4d+8$ in $\mathbb{C}^{d}\otimes \mathbb{C}^{d}\otimes \mathbb{C}^{2}$, and respectively proved that the UPBs are SUCPBs in A$|$BC and B$|$CA bipartitions when $d$ is odd and even. Two types of UPBs have been obtained in $\mathbb{C}^{d}\otimes \mathbb{C}^{2}\otimes \mathbb{C}^{2}$, which are SUCPBs in at most one bipartition. We have completely clarified the relationship between UPBs and SUCPBs for all possible cases. Moreover, we have put forward a UPB with smaller size $d^3-3d^2+1$ in $\mathbb{C}^{d}\otimes \mathbb{C}^{d}\otimes \mathbb{C}^{d}$ that is an SUCPB in every bipartition.

There are further interesting open questions left, for instance, the lower bound of a UPB that is an SUCPB in every bipartition, the existence of UPBs that are still UPBs in every bipartition, and UPBs that are strongly nonlocal in $N$-partite systems for $N\geq 5$. In the manuscript, we have focused on the fact that the SUCPB is inextricably linked to the quasi U-tile structure. The results may provide some theoretical and methodological reference to further investigations on the related topics such as intrinsic links between the UPB and the hypercube.

\begin{acknowledgments}
This work is supported by the Basic Research Project of Shijiazhuang Municipal Universities in Hebei Province (Grant No. 241790697A), the Natural Science Foundation of Hebei Province (Grant No. A2023205045), 
the NSFC (Grants Nos. 62272208, 12171044, 12075159, 12305030 and 12347104), the specific research fund of the Innovation Platform for Academicians of Hainan Province, the HKU Seed Fund for Basic Research for New Staff (No. 2201100596), Guangdong Natural Science Fund (No. 2023A1515012185), Hong Kong Research Grant Council (RGC) (No. 27300823, N\_HKU718/23, and R6010-23), Guangdong Provincial Quantum Science Strategic Initiative (No. GDZX2200001). 
\end{acknowledgments}

\section*{Data Availability Statement}

Data sharing is not applicable to this article as no new data were created or analyzed in this study.

\bibliography{reference}% Produces the bibliography via BibTeX.

\end{document}